\documentclass[acmsmall, sigconf, screen]{acmart}

\usepackage{amsfonts}
\usepackage{amsmath}
\usepackage{amsthm}

\usepackage{blindtext}
\usepackage{xpatch}
\usepackage{makecell}
\usepackage{multirow}
\usepackage{colortbl}

\usepackage[linesnumbered,ruled,vlined]{algorithm2e}
\SetKwInput{KwInput}{Input}                
\SetKwInput{KwOutput}{Output}              

\usepackage{mathtools}
\usepackage{enumitem}
\usepackage{graphicx}
\usepackage{stmaryrd}
\usepackage{mdframed}
\usepackage{hyperref}
\usepackage{xcolor}
\usepackage{listings}
\usepackage{tikz}
\usetikzlibrary{matrix,positioning}
\usepackage{fancyhdr} 
\usepackage{amsbsy}
\usepackage{comment}

\usepackage{subcaption}

\usepackage{pdfpages}
\usepackage{soul}

\usepackage{marvosym}
\usepackage{orcidlink}

\hypersetup{
  colorlinks=true,
  linkcolor=blue,
  filecolor=magenta,
  citecolor=violet,
  urlcolor=cyan,
  pdftitle={Overleaf Example},
  pdfpagemode=FullScreen
}

\definecolor{darkgreen}{rgb}{0.0, 0.3, 0.0}
\definecolor{codegreen}{rgb}{0,0.4,0}
\definecolor{codegray}{rgb}{0.5,0.5,0.5}
\definecolor{codepurple}{rgb}{0.58,0,0.82}
\definecolor{darkmagenta}{rgb}{0.55,0,0.55}
\definecolor{backcolour}{rgb}{1.0,1.0,1.0}
\definecolor{eminence}{RGB}{108,48,130}
\definecolor{burntorange}{rgb}{0.8, 0.33, 0.0}

\lstdefinestyle{assembly}{
  belowcaptionskip=1\baselineskip,
  breaklines=true,
  frame=single,
  numbers=left,
  numberstyle=\tiny\color{codegray},
  showstringspaces=false,
  commentstyle=\itshape\color{gray!40!black},
  basicstyle=\ttfamily,
  upquote=true,
  stepnumber=1,
  numbersep=8pt,
  aboveskip=1.5em,
  belowskip=1.5em,
  xleftmargin=2em,
  frame=single,
  framexleftmargin=.35em
}

\lstdefinelanguage{pseudo}{
  sensitive = true,
  keywords={operation, var, guard, if, else, for, type, record, fin},
  numbers=left,
  stepnumber=1,
  numbersep=8pt,
  showstringspaces=false,
  comment=[l]{//},
  morecomment=[s]{/*}{*/},
}
\lstdefinestyle{pseudostyle}{
    language=pseudo,
    commentstyle=\color{green!50!black},
    keywordstyle=\bf\ttfamily\color{orange},
    numberstyle=\tiny\color{codegray},
    stringstyle=\color{codepurple},
    breakatwhitespace=false,
    basicstyle=\footnotesize\ttfamily,
    breaklines=true,
    backgroundcolor=\color{blue!4},
    captionpos=b,
    keepspaces=true,
    numbers=left,
    numbersep=5pt,
    showspaces=false,
    showtabs=false,
    tabsize=2,
    frame=none,
    frame=shadowbox,
    xleftmargin=0.7cm,
    xrightmargin=0.2cm,
    escapeinside={<@}{@>}
}

\lstdefinestyle{baseC}{
  language=C,
basicstyle=\footnotesize\ttfamily,
keywordstyle=\color{blue}\ttfamily,
stringstyle=\color{red}\ttfamily,
commentstyle=\itshape\color{green!50!black},
morecomment=[l][\color{magenta}]{\#},
backgroundcolor=\color{gray!7},
xleftmargin=0.7cm,
xrightmargin=0.2cm,
frame=none,
  emptylines=1,
  numberstyle=\tiny\color{codegray},
  stringstyle=\color{codepurple},
  breaklines=true,
  captionpos=b,
  keepspaces=true,
  numbers=left,
  numbersep=5pt,
  showspaces=false,
  showstringspaces=false,
  showtabs=false,
  tabsize=2,
  frame=shadowbox,
}

\newtheorem{example}{Example}
\newtheorem{lemma}{Lemma}
\newtheorem{observation}{Observation}
\newtheorem{theorem}{Theorem}
\newtheorem{definition}{Definition}

\newcommand{\toolname}{\textsc{SECANT}}

\newcommand{\adepth}{d}
\newcommand{\agrammar}{G}
\newcommand{\rbuf}{\text{\pseudocb{reuse\_buf}}}

\newcommand{\srcvars}{\setvars_{\pubinp}}
\newcommand{\obsvars}{\setvars_{\mathsf{obs}}}

\newcommand{\askeleton}{w}
\newcommand{\apred}{\phi}

\newcommand{\fpop}{\texttt{MulOp}}
\newcommand{\ldop}{\texttt{LdOp}}
\newcommand{\datadep}{\mathsf{datadep}}
\newcommand{\addrdep}{\mathsf{addrdep}}
\newcommand{\sameaddr}{\mathsf{sameaddr}}
\newcommand{\diffaddr}{\mathsf{diffaddr}}
\newcommand{\speculates}{\mathsf{speculative}}

\newcommand{\lowoperands}{\mathsf{lowoperands}}
\newcommand{\highresult}{\mathsf{highresult}}
\newcommand{\lowresult}{\mathsf{lowresult}}
\newcommand{\destreg}{\mathsf{destreg}}
\newcommand{\srcaddr}{\mathsf{srcaddr}}
\newcommand{\srcdata}{\mathsf{srcdata}}
\newcommand{\areg}{\texttt{reg}}
\newcommand{\rsone}{\texttt{rs1}}
\newcommand{\rstwo}{\texttt{rs2}}
\newcommand{\rd}{\texttt{rd}}
\newcommand{\addr}{\texttt{addr}}

\newcommand{\regfile}{\mathsf{regfile}}

\newcommand{\aprogram}{C}

\newcommand{\adomain}{\mathbb{D}}

\newcommand{\afacet}{f}

\newcommand{\setvars}{\mathrm{V}}

\newcommand{\anopermap}{\omega}

\newcommand{\anassignment}{\sigma}
\newcommand{\setassignments}{\Sigma}

\newcommand{\initialpred}{\setassignments_\mathsf{init}}

\newcommand{\anop}{\mathsf{op}}
\newcommand{\aninstr}{\mathsf{inst}}

\newcommand{\setops}{\mathsf{Op}}
\newcommand{\setinstrs}{\mathsf{Inst}}

\newcommand{\nspec}{\mathsf{ns}}

\newcommand{\atrace}{\pi}

\newcommand{\arch}{\mathsf{a}}
\newcommand{\march}{\mathsf{m}}
\newcommand{\setarchvars}{\setvars_{\arch}}
\newcommand{\setmarchvars}{\setvars_{\march}}

\newcommand{\pubinp}{\mathsf{pub}}
\newcommand{\secinp}{\mathsf{sec}}
\newcommand{\obsout}{\mathsf{obs}}

\newcommand{\apattern}{\mathsf{p}}
\newcommand{\apredicate}{\phi}

\newcommand{\transrel}{\mathsf{T}}

\newcommand{\transrelsem}[2]{\llbracket #1 \rrbracket(#2)}
\newcommand{\transrelsemnspec}[2]{\llbracket #1 \rrbracket_{\nspec}(#2)}

\definecolor{clr-background}{RGB}{255,255,255}
\definecolor{clr-text}{RGB}{0,0,0}
\definecolor{clr-string}{RGB}{163,21,21}
\definecolor{clr-namespace}{RGB}{0,0,0}
\definecolor{clr-preprocessor}{RGB}{128,128,128}
\definecolor{clr-keyword}{RGB}{0,0,255}
\definecolor{clr-type}{RGB}{43,145,175}
\definecolor{clr-variable}{RGB}{0,0,0}
\definecolor{clr-constant}{RGB}{111,0,138} 
\definecolor{clr-comment}{RGB}{0,128,0}

\lstdefinestyle{VS2017}{
	backgroundcolor=\color{clr-background},
	basicstyle=\ttfamily\color{clr-text}, 
	stringstyle=\color{clr-string},
	identifierstyle=\color{clr-variable}, 
	commentstyle=\color{clr-comment},
	directivestyle=\color{clr-preprocessor}, 
	keywordstyle=\color{clr-type},
	keywordstyle={[2]\color{clr-constant}}, 
	tabsize=4
}

\newcommand{\semhyper}{\mathrm{SH}}
\newcommand{\attpat}{\mathrm{AP}}
\newcommand{\sempat}{\textrm{SemPat}}
\newcommand{\pseudocb}[1]{\text{\lstinline[style=pseudostyle, basicstyle=\ttfamily]{#1}}}

\newcommand{\baseCcb}[1]{\lstinline[style=baseC, basicstyle=\ttfamily]{#1}}

\newcommand{\platcompreuse}{\textsc{PlatCR}}

\newcommand{\platss}{\textsc{PlatSS}}

\newcommand{\platsynth}{\textsc{PlatSynth}}
\newcommand{\pdep}{\mathsf{pdep}}
\newcommand{\gdep}{\mathsf{gdep}}

\newcommand{\abuf}{\text{\pseudocb{buf}}}

\newcommand{\orcidlogo}[1]{\href{https://orcid.org/#1}{\mbox{\scalerel*{
\begin{tikzpicture}[yscale=-1,transform shape]
\pic{orcidlogo};
\end{tikzpicture}
}{|}}}}

\acmConference[]{ArXiv}{May 2024}{N.A.}

\renewcommand\footnotetextcopyrightpermission[1]{} 
\setcopyright{none}

\title{SemPat: Using Hyperproperty-based Semantic Analysis to Generate Microarchitectural Attack Patterns}

\author[A. Godbole]{Adwait Godbole \href{mailto:adwait@berkeley.edu}{$^{\text{(\Letter)}}$} 
\orcidlink{0000-0001-7704-304X}}
\email{adwait@berkeley.edu}
\affiliation{
    \institution{University of California, Berkeley}
    \city{Berkeley}
    \state{CA}
    \country{USA}
}

\author[Y. A. Manerkar]{Yatin A. Manerkar \orcidlink{0000-0002-6954-2292}}
\email{manerkar@umich.edu}
\affiliation{
    \institution{University of Michigan}
    \city{Ann Arbor}
    \state{MI}
    \country{USA}
}

\author[S. A. Seshia]{Sanjit A. Seshia \orcidlink{0000-0001-6190-8707}}
\email{sseshia@berkeley.edu}
\affiliation{
    \institution{University of California, Berkeley}
    \city{Berkeley}
    \state{CA}
    \country{USA}
}


\begin{document}

\begin{abstract}

Microarchitectural security verification 
of software has seen the emergence of two broad classes 
of approaches.
The first is based on \textit{semantic security properties} 
(e.g., non-interference) which are verified for a given program 
and a specified abstract model of the hardware microarchitecture.
The second is based on \textit{attack patterns}, which, if found in a 
program execution, indicates the presence of an exploit. 
While the former uses a formal specification that can capture several gadget variants targeting the same vulnerability, it is limited by the scalability of verification.
Patterns, while more scalable, must be currently constructed manually, as they are narrower in scope and sensitive to gadget-specific structure.

This work develops a technique that, given a non-interference-based semantic security hyperproperty, automatically generates attack patterns up to a certain complexity parameter (called the skeleton size).
Thus, we combine the advantages of both approaches: security can be specified by a hyperproperty that uniformly captures several gadget variants, while automatically generated patterns can be used for scalable verification.
We implement our approach in a tool and demonstrate the ability to generate new patterns, (e.g., for SpectreV1, SpectreV4) and improved scalability using the generated patterns over hyperproperty-based verification.

\end{abstract}
\maketitle

\section{Introduction}

Modern processors are packed with performance-improving 
microarchitectural mechanisms such as caches, 
speculation, prefetching.
These lead to subtle microarchitectural-level interactions 
that can be exploited by hardware execution attacks  
to leak sensitive data from a victim
(e.g., \cite{kocher, Kocher2019SpectreAE, Lipp2018MeltdownRK, Canella2019FalloutLD, Schwarz2019ZombieLoadCD}).
Recently, several techniques have been proposed to detect the presence of 
such vulnerabilities in software and/or verify their absence 
(e.g. \cite{Cheang2019AFA, Guarnieri2020SpectectorPD, Balliu2020InSpectreBA, Mosier2022AxiomaticHC, Len2021CatsVS, Daniel2019BinsecRelER,Daniel2021HuntingTH}).
While sharing the goal of security verification, these 
approaches adopt different (hardware) platform models,
security specifications, and have different strengths.

We observe the emergence of two broad classes of approaches
based on the security specification they adopt:
\textit{semantic hyperproperty} ($\semhyper$) based approaches
(e.g., \cite{Cheang2019AFA, Guarnieri2020SpectectorPD, 
Balliu2020InSpectreBA, Guarnieri2021HardwareSoftwareCF, 
Fabian2022AutomaticDO}) perform verification with respect to
a hyperproperty-based \cite{Clarkson2008Hyperproperties} security specification, while \textit{attack pattern} ($\attpat$) 
based approaches 
(e.g., \cite{Mosier2022AxiomaticHC, Len2021CatsVS, Trippel2019SecurityVV})
perform verification using patterns that indicate the 
existence of an exploit.
$\semhyper$ approaches provide 
advantages of uniform specification and formal guarantees,
while $\attpat$ approaches have better scalability.
In this work, \textbf{we enable automated generation 
of attack patterns given a hyperproperty-based specification},
thereby combining their strengths.

\text{$\semhyper$ approaches}
(e.g., \cite{Cheang2019AFA, Guarnieri2020SpectectorPD, Balliu2020InSpectreBA, 
Guarnieri2021HardwareSoftwareCF, Fabian2022AutomaticDO})
identify a transition system-based platform model
and define program security \textit{semantically} as a hyperproperty
(e.g., non-interference (NI) \cite{Goguen1984UnwindingAI, Clarkson2008Hyperproperties})
over executions of this transition system.
This hyperproperty is verified using approaches such as model checking 
\cite{Clarke1993ModelCA} or symbolic execution \cite{King1976SymbolicEA}.
By semantically characterizing vulnerabilities, 
hyperproperties
allow \textit{uniform specification}: i.e., specification 
that captures several exploit gadgets that target the same vulnerability,
while differing in syntactic structure.
This results in strong, high-coverage security guarantees. 
However, these approaches are often limited by scalability owing to microarchitectural platform model complexity.

\text{$\attpat$ approaches}
(e.g. \cite{Mosier2022AxiomaticHC, Len2021CatsVS, Trippel2019SecurityVV}) 
use \textit{attack patterns} to detect vulnerabilities.
Patterns identify execution fragments that are indicative of an exploit;
program executions embedding these patterns are flagged as vulnerable.
Patterns can be defined (and checked) over a platform 
model that is more abstract than the microarchitectural models 
used by $\semhyper$ approaches.
This abstraction leads to simpler verification queries 
which scale better with program size.
However, each pattern captures \textit{specific execution scenarios} 
and is sensitive to structural variability in exploit gadgets, 
even when the gadgets target the same microarchitectural feature. 
Covering the attack vector requires an
enumeration of all possible patterns that it encompasses.
This
can be tedious to perform manually, leading to incomplete
coverage.

\textit{$\sempat$.}
The above discussion 
motivates the question:
\textit{Can we combine the scalability of $\attpat$ approaches 
with the uniform specification and guarantees of 
$\semhyper$ approaches?}
This work aims to resolve this question in the context of 
microarchitectural vulnerabilities
by developing algorithms to generate attack 
patterns given a platform model and a non-interference-based hyperproperty.
\textbf{Our key insight} is to use the hyperproperty 
as the specification against which to check 
whether a candidate pattern represents an exploit. 
Our technique (\S\ref{sec:approach}) ensures that the generated patterns 
capture all executions in which there is a dependency-closed 
sub-execution of size $k$ that violates the hyperproperty.
This property, termed as $k$-completeness (Eq. \ref{eq:completeness}), enables
combining formal specification via hyperproperties with 
scalable verification via patterns.

Our pattern generation is guided by a \textit{grammar} 
(\S\ref{subsec:patgrammar}) that identifies the space 
of constraints that patterns are defined over.
The choice of grammar captures the tradeoff between generality and the 
precision (false positives) of the generated patterns.
Specialized grammars result in patterns which have fewer false positives
but are microarchitectural-implementation dependent.
Our approach traverses the grammar-induced search space, 
checking candidates against the hyperproperty, a la grammar-based synthesis 
(\cite{Alur2013SyntaxguidedS,Jha2015ATO}).

\textbf{Contributions.} Our contributions are as follows:
\begin{enumerate}
    \item \textbf{Formally relating semantic hyperproperty and attack pattern based approaches:} 
    We compare and formally relate $\semhyper$
    and $\attpat$ approaches to exploit detection. 
    Based on the insight that hyperproperties can serve as a 
    specification for patterns, we propose the problem of 
    generating the latter given the former.
    \item \textbf{$k$-complete automatic pattern generation:} 
    We develop a grammar-based search algorithm for automatic pattern 
    generation that ensures that the generated patterns capture all 
    non-interference violations up to a certain complexity parameter $k$ 
    (termed skeleton size).
    \item \textbf{Implementation and Evaluation - new patterns and improved verification performance:} 
    We develop a prototype tool implementing our pattern generation 
    technique and for verifying software binaries using the generated patterns.
    We evaluate our approach by generating attack patterns
    and using them to analyze variants of Spectre-style exploits.
    We demonstrate: (a) the ability to generate previously unknown 
    patterns for existing (e.g., Spectre-BCB, Spectre-STL) 
    attacks as well as variants targetting alternative (e.g. computation-unit-based)
    side-channels, and (b) upto 2 orders-of-magnitude performance improvement 
    on litmus tests with better scaling using 
    the generated patterns compared to hyperproperty-based analysis.
\end{enumerate}

\textbf{Outline.}
In \S\ref{sec:motivation} we provide background on 
$\semhyper$ and $\attpat$ and motivate our contribution.
We provide our programming model in \S\ref{sec:progmodel},
followed by the problem 
formulation in \S\ref{sec:problemform}.
In \S\ref{sec:approach} we describe SemPat, our approach to
generate attack patterns based on a semantic platform model and 
security hyperproperty.
We describe the evaluation methodology in \S\ref{sec:methodology} and 
present experimental results in \S\ref{sec:evaluation}.
We discuss our approach, and its limitations in \S\ref{sec:discussion},
related work in \S\ref{sec:related}, and 
\S\ref{sec:conclusion} concludes.
We provide proofs/algorithm pseudocode in the supplementary material appendix.

\section{Background and Motivation}
\label{sec:motivation}

\subsection{Microarchitectural Execution Attacks}

We provide background on microarchitectural execution attacks
using the Spectre Bounds Check Bypass (BCB) vulnerability.
We refer the reader to literature (e.g., \cite{Kocher2019SpectreAE, Ge2018ASO, Canella2019FalloutLD})
for extensive surveys.

\subsubsection{Spectre-BCB}
The \baseCcb{victimA} function in Fig. \ref{fig:programexamples} 
shows the Spectre-BCB (also known as SpectreV1) vulnerability gadget.
In Spectre-BCB, an attacker induces an invocation of the \baseCcb{victimA} function 
with an argument \baseCcb{i} which is out-of-bounds of array \baseCcb{arr1}.
While architecturally this is an illegal access, branch mis-speculation 
can allow \baseCcb{arr1[i]}, followed by \baseCcb{arr2[arr1[i] << CL_INDEX]} 
to be accessed \textit{speculatively}.
In particular, the second load leaves a residue in the cache which 
is a function of the accessed address (i.e., \baseCcb{arr2+arr1[i] << CL_INDEX}).
This residue is preserved even after speculative rollback, and can be 
observed by the attacker (e.g., using a Prime+Probe primitive 
\cite{Liu2015LastLevelCS}).
Thus the attacker can observe the value of \baseCcb{arr1[i]} for an
out-of-bounds index \baseCcb{i}, leading to a security vulnerability.
The key aspect of this vulnerability is that it leverages \textit{microarchitectural
features} - branch (mis)-speculation and cache side-channels -
to leak information.

\subsubsection{Variant Vulnerabilities}

Over the years several variants of microarchitectural hardware execution 
attacks have been demonstrated.
These vary both at the software level (e.g., the structure of the 
vulnerability gadget occurring in the program library/executable)
as well as at the microarchitectural level (e.g., the underlying 
speculative mechanism/side-channel that is exploited).

The function \baseCcb{victimC} in Fig. \ref{fig:programexamples} 
shows a variant of \baseCcb{victimA}, that replaces a cache-based
side channel with a computation unit-based side channel.
Like \baseCcb{victimA}, \baseCcb{victimC} also exploits branch speculation to 
perform an out-of-bounds load of \baseCcb{arr1fp[i]} (a floating-point array).
The loaded value is used in a floating-point multiply operation.
If the microarchitectural implementation of this operation
has data-operand dependent timing (e.g., \cite{Vicarte2021OpeningPB,Ge2018ASO})
then an execution time measurement leaks the value of \baseCcb{arr1fp[i]}.
Computation units have been targeted in this manner by previously demonstrated exploits (e.g., NetSpectre \cite{Schwarz2019NetSpectreRA} 
uses AVX-based timing side channels) as well as conjectured vulnerabilities 
(e.g., \cite{Vicarte2021OpeningPB}).

\begin{figure*}
    \centering
    \includegraphics[scale=0.52]{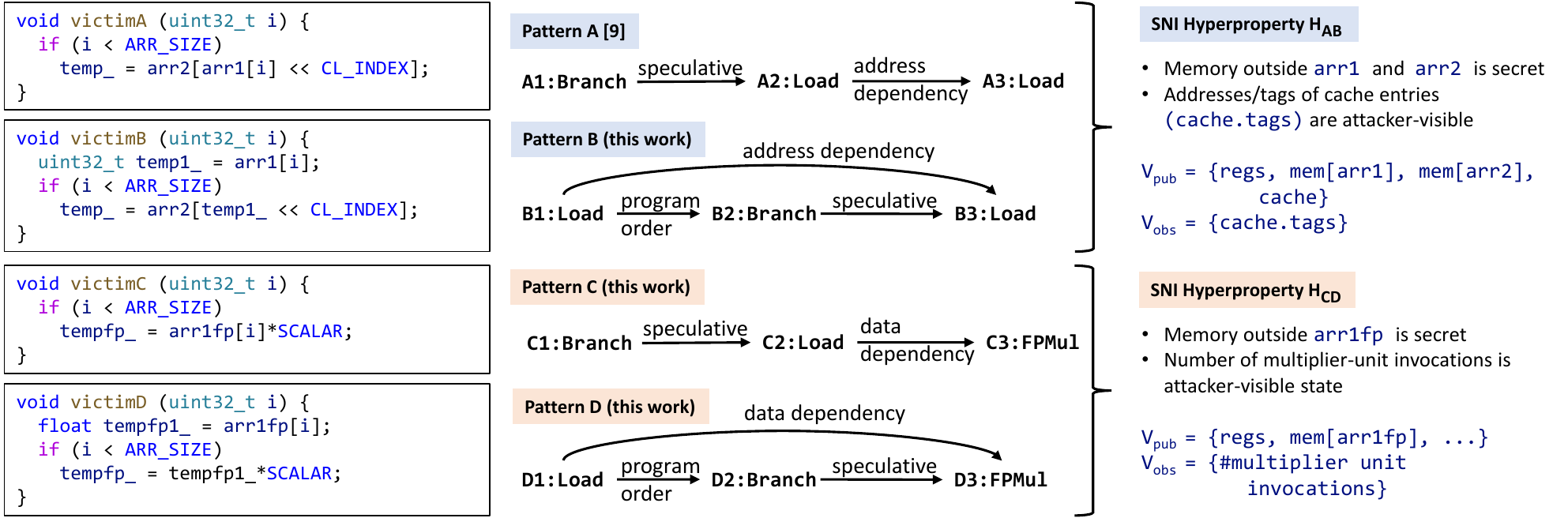}
    \caption{\textbf{Left}: Variants of speculative exploits targeting branch (mis)-speculation
    with cache-based (\baseCcb{victimA}, \baseCcb{victimB}) and computation unit-based (\baseCcb{victimC}, \baseCcb{victimD}) side channels.
    \textbf{Centre}: While attack patterns that can detect these exploits, variants of the same 
    (Spectre-BCB) vulnerability - \baseCcb{victimA} and \baseCcb{victimB} - require different patterns, 
    as do \baseCcb{victimC} and \baseCcb{victimD}.
    In (A, C), the load \baseCcb{arr1[i]} is performed within the speculative window 
    (after the branch), while in (B, D) it is before the branch.
    \textbf{Right}: Unlike attack patterns, a semantic hyperproperty \textit{uniformly characterizes} several exploits 
    aimed a particular microarchitectural vulnerability. The Speculative Non-Interference \cite{Guarnieri2020SpectectorPD} (SNI) 
    Hyperproperty $\text{H}_{\text{AB}}$ can identify both \baseCcb{victimA} and \baseCcb{victimB} 
    (as well as others) as exploits, while the SNI Hyperproperty $\text{H}_{\text{CD}}$ identifies both \baseCcb{victimC} and \baseCcb{victimD}. 
    While the hyperproperty changes with the platform's microarchitectural features 
    (e.g., speculation and side-channels), it is robust to variances in the exploit (program) structure itself.}
    \label{fig:programexamples}
\end{figure*}

\subsection{Analyzing Software for Vulnerabilities}

Analyzing programs for the existence of microarchitectural
vulnerabilities is an important yet challenging problem.

\subsubsection{$\attpat$: Attack Pattern-based Analysis}
\label{subsubsec:attpat}

One family of approaches, which we call $\attpat$ approaches, 
detect vulnerabilities using \textit{attack patterns}.
An attack pattern is a small execution fragment
which, if \textit{embedded} in some (larger) program execution,
indicates the presence of a vulnerability.
Pattern A (Fig. \ref{fig:programexamples}) illustrates one such pattern
from existing work \cite{Mosier2022AxiomaticHC}.
Pattern A matches executions where a load (A2) is speculatively executed
following a branch (A1), and where the address of a subsequent load (A3) 
depends on the value loaded by A2.
Fig. \ref{fig:patternmatch} (Match A) shows the compiled \baseCcb{victimA}, 
and depicts how the instruction nodes A1, A2, A3 
in pattern A match instructions in the binary execution.

\begin{table}
    \centering
\resizebox*{0.45\textwidth}{!}{
    \begin{tabular}{c|c|c}
        \hline
        $\attpat$-approach & Execution model & Attack pattern variant \\ \hline \hline
        CheckMate \cite{Trippel2019SecurityVV} & $\mu$hb graphs \cite{Lustig2016COATCheckVM} & \textit{Bad execution patterns} \\  \hline
        \makecell{Cats vs. Spectre \\ \cite{Len2021CatsVS}} & \makecell{\texttt{cat} \cite{Alglave2014HerdingCM} based \\ event structures} & $\mathsf{sec} \rightarrow^* \textit{\textbf{Obs}}$ paths \\ \hline
        \makecell{Axiomatic HW/SW \\ Contracts \cite{Mosier2022AxiomaticHC}} & \makecell{Leakage containment \\ models (LCMs)} & Transmitters \\ \hline 
    \end{tabular}
}
    \vspace{1em}
    \caption{Examples of execution modelling choices and 
    pattern variants used in some $\attpat$-approaches.}
    \label{tab:attpat-approaches}
\end{table}

\subsubsection*{$\attpat$ approach variants}
Existing $\attpat$-based detection approaches identify a set of such patterns, 
and then analyze programs for the existence of pattern embeddings.
These approaches model executions and define patterns either at the
architectural \cite{Mosier2022AxiomaticHC,Len2021CatsVS} or the microarchitectural 
\cite{Trippel2019SecurityVV} level.
Architecture-level approaches are augmented with sufficient
microarchitectural detail to capture the vulnerability-specific features
(e.g. \textit{xstate} from Mosier et. al. \cite{Mosier2022AxiomaticHC}).
Table \ref{tab:attpat-approaches} provides a summary of these differences.

Patterns defined at the architectural level \textit{abstract away}
complex microarchitectural details of the platform.
This results in simpler verification queries,
allowing analysis to scale to larger programs (e.g., see~\cite{Mosier2022AxiomaticHC,Len2021CatsVS}).

\begin{figure}
    \centering
    \includegraphics[scale=0.55]{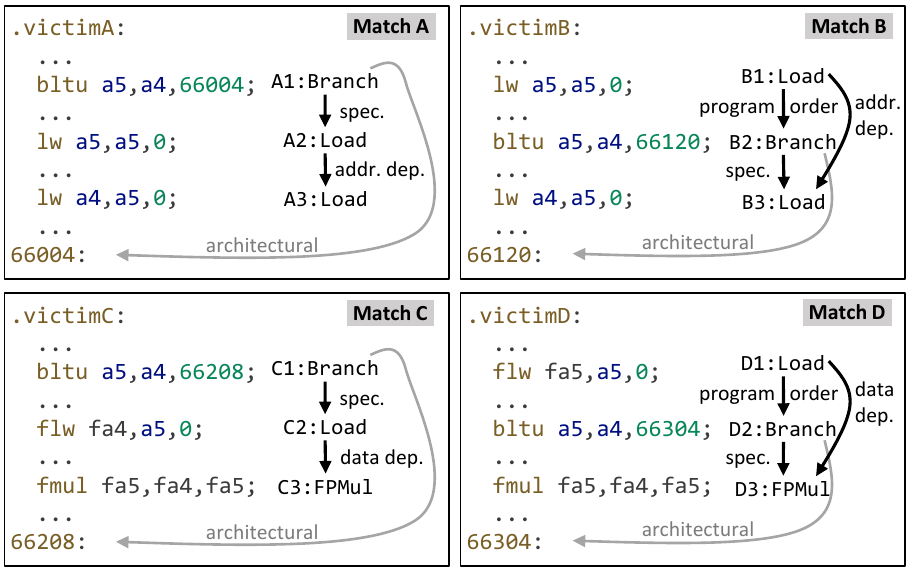}
    \caption{Patterns A, B, C, D matched against executions of programs \baseCcb{victimA},\baseCcb{victimB},\baseCcb{victimC},\baseCcb{victimD} respectively.}
    \label{fig:patternmatch}
\end{figure}

\subsubsection*{Patterns are tied to gadget-specific structure
\footnote{Which is why we call these \textit{attack} patterns, since they only capture
specific attack executions, and not the common underlying vulnerability.}}

While attack pattern-based analysis has the benefit of greater scalability,
patterns do not generalize well to other gadget variants that exploit the 
same underlying microarchitectural features.
To illustrate this, function \baseCcb{victimB} in Fig. \ref{fig:programexamples} 
is a modified version of \baseCcb{victimA}, where the load \baseCcb{arr1[i]} 
is performed non-speculatively, i.e. before the branch.
Due to this load$-$branch inversion, pattern A does not match the execution 
of \baseCcb{victimB}, even though both \baseCcb{victimA} and \baseCcb{victimB}
exploit the same (Spectre-BCB) mechanism.
In Fig. \ref{fig:programexamples} we provide pattern B that reorders 
the branch with the first load, and thus 
\textit{can} capture \baseCcb{victimB} (depicted in Fig. \ref{fig:patternmatch} Match B).
This highlights the fact that attack patterns are tied to gadget-specific structure.

\begin{table}
    \centering
\resizebox*{0.47\textwidth}{!}{
    \begin{tabular}{c|c}
        \hline
        $\semhyper$-approach & Non-interference (NI) variant  \\ \hline \hline
        \makecell{A Formal Approach to \\ to Secure Speculation \cite{Cheang2019AFA}}  & \makecell{Trace-property dependent \\ observational-determinism} \\  \hline
        Spectector \cite{Guarnieri2020SpectectorPD} & Speculative non-interference \\ \hline
        InSpectre \cite{Balliu2020InSpectreBA} & Conditional non-interference \\ \hline 
        \rowcolor{gray!15}
        Hardware-software contracts \cite{Guarnieri2021HardwareSoftwareCF} & Contract conditioned NI \\ \hline
        \rowcolor{gray!15}
        \cellcolor{gray!15}\makecell{Automated detection of \\ speculative attack combinations \cite{Fabian2022AutomaticDO}}  & Speculative NI  \\ \hline
    \end{tabular}
}
    \vspace{1em}
    \caption{Summary of non-interference variants used in some $\semhyper$ approaches.}
    \label{tab:semhyper-approaches}
\end{table}

\subsubsection{$\semhyper$: Semantic Hyperproperty-based Analysis}
\label{subsubsec:semhyper}

Non-interference-based hyperproperties \cite{Clarkson2008Hyperproperties} 
allow security specifications such as confidentiality
(``secret variables should not affect public outputs'') and integrity
(``public inputs should not affect protected variables'').
$\semhyper$-approaches, e.g., \cite{Cheang2019AFA, Guarnieri2020SpectectorPD, Balliu2020InSpectreBA},
formulate security semantically using hyperproperties
and verify them against a platform model.

\subsubsection*{Hyperproperties allow uniform security specification}
The hyperproperty $\text{H}_\text{AB}$ (Fig. \ref{fig:programexamples} (right)) 
specifies the memory \textit{outside} (public) arrays 
\baseCcb{arr1}, \baseCcb{arr2} as being secret, and the cache tags as being public output.
This captures a cache-based side channel where the tags
are attacker-observable (e.g., using Prime+Probe \cite{Liu2015LastLevelCS}).
\baseCcb{victimA} violates $\text{H}_\text{AB}$, since the cache state
is tainted with the (speculatively loaded) address of 
the second load which is a value outside \baseCcb{arr1}, \baseCcb{arr2}.
Although it has a different structure than \baseCcb{victimA},
\baseCcb{victimB} also violates $\text{H}_\text{AB}$
as it too leaves a cache residue.
Thus, unlike attack patterns, hyperproperties such as $\text{H}_\text{AB}$ are agnostic to gadget-specific structure.
By uniformly capturing several exploit-variants targetting a particular 
microarchitectural vulnerability, $\semhyper$-approaches provide wide-scoped, strong 
guarantees.

\subsubsection*{Hyperproperty verification}
$\semhyper$-approaches base their analysis on variants of
non-interference, which we summarize in Tab. \ref{tab:semhyper-approaches},
and discuss in more detail in \S\ref{subsec:securitydef}.
They verify this hyperproperty against a 
\textit{semantic platform model}, which
identifies (a) system state (variables in the system),
and (b) execution semantics of operations.
Fig. \ref{fig:simpleplatformmodel} illustrates a
platform fragment with register file, memory and cache state,
and \pseudocb{load} and \pseudocb{alu} operations.
Non-interference-based hyperproperties 
can be converted into a single trace property by performing \textit{self-composition},
i.e., composing together copies of the platform
(\cite{Clarkson2008Hyperproperties,Terauchi2005SecureIF}).
This single trace property can be checked 
by invoking a model checking or software verification procedure.

Self-composition-based verification against a semantic model
has two drawbacks.
Firstly, microarchitectural detail in the platform model
results in large verification queries.
(e.g., a query for $\text{H}_\text{AB}$ over the 
platform from Fig. \ref{tab:platformstate} would have to 
encode the cache state).
In comparison, architectural level pattern-based queries
(e.g., with pattern A) are smaller, enabling faster verification.
Secondly, self-composition results in a further increase in the state-space
(and query) and consequentially adversely affects performance.

\begin{figure}
    \centering
\begin{lstlisting}[style=pseudostyle]
// System state
var regs  : [regindex_t]word_t
var cData : [index_t]word_t
var cTag  : [index_t]tag_t
var mem   : [word_t]word_t
// Load operation
operation load (rs, rd, imm) {
  addr = regs[rs] + imm;  // Compute the address
  if (cacheTag(addr_to_tag(addr)) == addr) {
    ... // Load from cache if hit
  } else {
    data = mem[addr]; // Load from memory
  }
  regs[rd] = data; // Register writeback
}

// Generic ALU Register-Register operation
operation alu (rs1, rs2, rd, op) {
  if (op == ADD) regs[rd] = regs[rs1] + regs[rs2];
  ...
}
\end{lstlisting}
\caption{Fragment of a platform model with 
state variables and \pseudocb{load} and \pseudocb{alu} operation semantics.}
\label{fig:simpleplatformmodel}
\end{figure}

\subsection{Why convert from $\semhyper$ to $\attpat$?}
\label{subsec:cando}

\subsubsection*{Manual pattern generation is error-prone}

While more scalable, $\attpat$ approaches require creation of 
several patterns, due to their gadget-specificity.
To avoid unsound analysis (e.g., using only pattern A on \baseCcb{victimB}),
it is important that patterns are not missed, e.g., pattern B which we have not observed being formulated previously.

Moreover, patterns need to be recreated for newer microarchitectures (with newer vulnerabilities).
Consider \baseCcb{victimC} (Fig. \ref{fig:programexamples} left) which replaces a cache-based 
side channel (as in \baseCcb{victimA}) with a computation unit-based channel.
Its variant \baseCcb{victimD} inverts the first load and the branch (as in \baseCcb{victimB}).
These examples are inspired from \cite{Vicarte2021OpeningPB}, which 
hypothesizes the existence of computation-unit based side-channels on 
microarchitectures with data-operand-dependent timing 
\cite{IntelDOIT,Ge2018ASO,Vicarte2021OpeningPB}.
Existing work, which targets cache-based side channels, misses patterns C and D (Fig. \ref{fig:programexamples} center) 
that capture \baseCcb{victimC} and \baseCcb{victimD}.
Automating pattern generation can make patterns more comprehensive and cover newer microarchitectural features.

\subsubsection*{Semantic hyperproperties as a specification for automated pattern generation}

In this work, we propose a technique to automatically generate patterns for a 
given hyperproperty and microarchitecture. 
As an example we were able to automatically generate patterns C and D
from the (shared) hyperproperty $\text{H}_\text{CD}$.

\textbf{Our key insight} is using the semantic hyperproperty as a specification to guide 
pattern generation.
Since a hyperproperty can capture an entire class of exploits targetting a vulnerability, 
we can use it to determine whether a given pattern yields an exploit
by checking it against the hyperproperty. 
By automatically checking several candidates, 
we can identify a comprehensive set of patterns for that hyperproperty.
Thus, our technique replaces manual pattern creation with 
the requirement of specifying a hyperproperty and platform model.
To summarize, by developing an automated conversion technique 
from $\semhyper$-specifications to $\attpat$-based patterns,
\textbf{we combine the low-overhead, uniform specifications and formal guarantees of 
$\semhyper$ with the superior verification scalability of $\attpat$ to get the best of both 
worlds.}

\newcommand{\specvar}{\mathit{spec}}

\section{System Model}
\label{sec:progmodel}

In this section, we introduce our formal model for hardware 
platforms, which we later use to develop our 
problem formulation (\S\ref{sec:problemform}).
At a high level, the hardware platform is an
operational transition system over architectural 
and microarchitectural state variables.
Instructions executed on the platform induce transitions 
over this state.
We summarize these elements in Table \ref{tab:platformstate}.

\subsubsection{State}

The platform consists of variables $\setvars$
which take values from a domain $\adomain$.
$\setvars$ includes architectural 
($\setarchvars$) 
and microarchitectural 
($\setmarchvars$) 
variables.
The platform state is an assignment to these variables, 
$\anassignment: \setvars \rightarrow \adomain$.
We denote the set of all assignments as $\setassignments = \setvars \rightarrow \adomain$.

\begin{table}[t]
        \renewcommand{\arraystretch}{1.5}
    \centering
    \begin{tabular}{rrl}
    \makecell[r]{Variables (arch. and march.)}  & $v :$ & $\setvars =
    \setarchvars \cup \setmarchvars$  \\[0.1cm]
    \makecell[r]{States \\ (variable assignments)} & 
        $\sigma :$ 
    & 
        $\Sigma = \setvars \rightarrow \adomain$ 
    \\[0.1cm] \hline
    Operation code & $\anop :$ & $\setops$  \\
    Instructions & $\aninstr :$ & $\setinstrs = \{\anop(\anopermap)\}_{\anop, \anopermap}$ \\ \hline
    Full (speculative) semantics & $\transrel$ :& $\setinstrs \times \Sigma \rightarrow  \Sigma$ \\
    Non-speculative semantics & $\transrel_\nspec$ :& $\setinstrs \times \Sigma \rightarrow  \Sigma$ \\
    \end{tabular}
    \vspace{0.2em}
    \caption{Platform state and operational semantics.}
    \label{tab:platformstate}
\end{table}

\subsubsection{Instruction semantics}

The platform executes a set of instructions
of form $\aninstr = \anop(\anopermap)$,
where $\anop$ is the instruction opcode, 
and $\anopermap$ are operands.
The platform assigns two \textit{transition semantics}
to each instruction:
a full semantics (allowing speculation) denoted as $\transrel$
and
a non-speculative semantics denoted as $\transrel_\nspec$.
Both semantics can be viewed as functions taking an 
instruction and the current platform state as input and returning the next 
platform state (obtained after executing the instruction).
The full semantics defines the behaviour when
the platform \textit{can} speculate (not necessarily enforcing speculative behaviour at all times)
while the non-speculative semantics
defines behaviours when speculation is disabled.

\begin{example}
\label{ex:formalmodel}
For the platform model in Fig. \ref{fig:simpleplatformmodel},
the architectural and microarchitectural variables are 
$\setarchvars = \{\pseudocb{mem}, \pseudocb{regs}\}$ and
$\setmarchvars = \{\pseudocb{cacheTag}, \pseudocb{cacheData}\}$,
with $\setvars = \setarchvars \cup \setmarchvars$.
The semantics of the load instruction 
updates the register file \pseudocb{regs}
and the cache state variables (\pseudocb{cacheData}, \pseudocb{cacheTag}) as
defined in Fig. \ref{fig:simpleplatformmodel}.
\end{example}

\subsubsection{Modelling speculation}
\label{subsubsec:specmodel}
Instructions, e.g., branches or loads (store-to-load forwarding), signal that they are 
initiating speculation by setting a variable $\specvar \in \setvars$.
Internally, the full semantics $\transrel$ defines instruction behaviour
by conditioning on $\specvar$: $\specvar$ being set
implies that the platform is \textit{currently} speculating.
Indeed, $\specvar$ can be set only in the full semantics ($\transrel$)
and not in $\transrel_\nspec$.
Prior work considers specialized semantics
that define when $\specvar$ is set
(e.g., oracle-based semantics \cite{Guarnieri2020SpectectorPD}).
Since we adopt a hardware-oriented model, we assume 
that this is explicitly defined in the transition semantics
of the speculating instruction.
Despeculation is assumed to be similarly defined 
(this time, however, by unsetting $\specvar$).

In our current implementation, we restrict speculation
to a single frame, and do not support nested speculation 
(e.g., speculative loads within a branch speculation context).
However, this is not a fundamental limitation of our approach;
extension to nested speculation is possible by 
defining a stack of frames storing architectural state.

\subsubsection{Executions}

The platform consumes a stream of instructions,
and transitions on them, thereby producing a 
trace of states.
Then, an input instruction stream 
$\aprogram = \aninstr_0, \aninstr_1, \ldots, \aninstr_n$,
starting in state $\anassignment_0$
leads to a sequence of states
$\atrace = \anassignment_0 ~ \anassignment_1 ~ \ldots ~ \anassignment_{n+1}$,
where $\anassignment_{i+1} = \transrel(\aninstr_i, \anassignment_i)$
(under the full semantics) and 
$\anassignment_{i+1} = \transrel_\nspec(\aninstr_i, \anassignment_i)$
(under the non-speculative semantics).
The execution generated by instruction stream $\aprogram$ 
from initial state $\anassignment$
is denoted as $\transrelsem{\aprogram}{\anassignment}$.
We similarly define non-speculative executions
$\transrelsemnspec{\aprogram}{\anassignment}$ in which instructions 
follow the $\transrel_\nspec$ transition relation.

\section{Specifications and Problem Formulation}
\label{sec:problemform}

In this section we formalize hyperproperty specifications (\S\ref{subsec:securitydef})
and the notion of attack patterns
that our approach aims to generate (\S\ref{subsec:transmit}).
We then discuss a technical limitation of pattern-based approaches in \S\ref{subsec:largeskel},
and formulate the problem statement in \S\ref{subsec:problemstmt}.

\subsection{Hyperproperty-based Security Specification}
\label{subsec:securitydef}

We follow existing work (\cite{Clarkson2008Hyperproperties,Guarnieri2020SpectectorPD}) 
to formalize non-interference-based security specifications.

\textbf{Non-interference}
(\cite{Clarkson2008Hyperproperties}) 
states that any pair of executions 
which begin in states with equivalent values of 
\textit{public (non-secret) variables} ($\setvars_\pubinp$) 
continue to have states with equivalent values of \textit{observable variables} ($\setvars_\obsout$):
\footnote{
Here, $\anassignment_1 \equiv_{\setvars'} \anassignment_2$ for $\setvars' \subseteq \setvars$ means that 
$\anassignment_1(v) = \anassignment_2(v)$ for all $v \in \setvars'$ (the assignments agree on all 
variables in $\setvars'$). For traces $\atrace_1$ and $\atrace_2$, $\atrace_1 \equiv_{\setvars'} \atrace_2$
holds if $\atrace_1[i] \equiv_{\setvars'} \atrace_2[i]$ for all $i$.
}
\begin{align*}
    \aprogram \models &~\mathsf{NI}(\initialpred, \setvars_{\pubinp}, \setvars_\obsout) \overset{\Delta}{=} \forall \anassignment_1, \anassignment_2 \in \initialpred. \\
    &\anassignment_1 \equiv_{\setvars_{\pubinp}} \anassignment_2 \implies \transrelsem{\aprogram}{\anassignment_1} \equiv_{\setvars_{\obsout}} \transrelsem{\aprogram}{\anassignment_2}
\end{align*}
This property is parameterized by the choice of $\initialpred$, $\setvars_\pubinp$ 
and $\setvars_\obsout$ and intuitively says that the observable variables are not affected
by the secret ($\setvars_\secinp = \setvars\setminus\setvars_\pubinp$) variables.
Non-interference expresses security against an attacker
that tries to infer the values of $\setvars_\secinp$  
by observing $\setvars_\obsout$.

\textbf{Speculative non-interference}
enforces non-interference 
\textit{only if} the program is non-interfering under non-speculative semantics:
\begin{align*}
    \aprogram \models &~\mathsf{SNI}(\initialpred, \setvars_\pubinp, \setvars_\obsout) \overset{\Delta}{=} \forall \anassignment_1, \anassignment_2 \in \initialpred. \\
    &(\anassignment_1 \equiv_{\setvars_{\pubinp}} \anassignment_2 \land \llbracket \aprogram \rrbracket_{\nspec}(\anassignment_1) \equiv_{\setvars_{\obsout}}  \llbracket \aprogram \rrbracket_{\nspec}(\anassignment_2)) \implies \\
    &\transrelsem{\aprogram}{\anassignment_1} \equiv_{\setvars_{\obsout}} \transrelsem{\aprogram}{\anassignment_2}
\end{align*}
Intuitively, SNI restricts the scope of non-interference enforcement,
we refer the reader to \cite{Guarnieri2020SpectectorPD, Cheang2019AFA} for more details.

Conditional/contract-based non-interference \cite{Guarnieri2021HardwareSoftwareCF} is another variant
which also restricts the scope of non-interference using 
an architectural semantics. It requires that a program
be non-interfering in the full semantics if it is non-interfering
in the architectural semantics. 
While we focus our presentation on non-interference
our technique also applies to these variants,
as demonstrated in \S\ref{sec:evaluation}.

\subsection{Attack Pattern-based Security}
\label{subsec:transmit}

\subsubsection{Patterns}

A pattern $\apattern$ is a pair $(\askeleton, \apred)$ 
of a template $\askeleton$ and a constraint ($\apred$), i.e., a boolean formula.
The template is a sequence over opcodes
$\askeleton = \anop_0 \cdot \anop_1 \cdots \anop_k$
that \textit{structurally} restricts executions the 
pattern can be embedded in, to those with an opcode subsequence matching the template.
The constraint ($\apred$) further \textit{semantically} 
constrains matching executions to those satisfying it.

\begin{figure}
    \centering
    \includegraphics[width=0.9\linewidth]{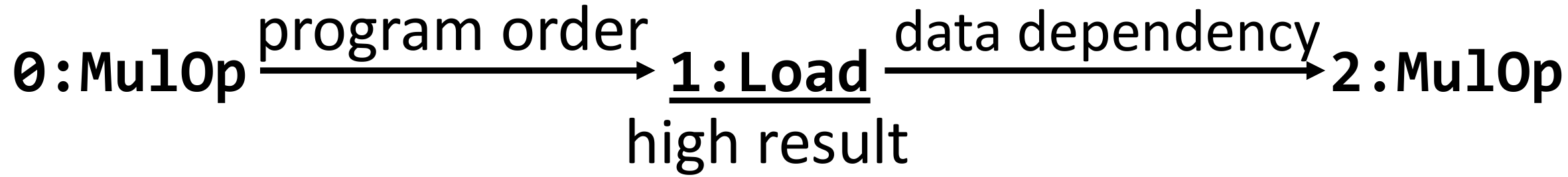}
    \caption{A pattern for a computation-based side channel.}
    \label{fig:patternexample}
\end{figure}

\begin{example}[$\fpop - \ldop - \fpop$ pattern]
The pattern from Fig. \ref{fig:patternexample} is formalized as $(\askeleton, \apredicate)$ where,
$\askeleton = (0:\fpop) \cdot (1:\ldop) \cdot (2:\fpop)$
and the constraint is $\phi \equiv \datadep(1:\ldop, 2:\fpop) \land 
\highresult(1: \ldop)$, i.e., there is a 
data dependency between the load and the second multiplication operation,
and a the loaded result is high (secret dependent).
Intuitively, this pattern matches executions satisfying $\phi$ 
with a $\fpop - \ldop - \fpop$ instruction subsequence.
\label{ex:pattern}
\end{example}

\subsubsection{Execution-embedding}

Now we formalize when a pattern embeds (matches) an execution,
which determines which programs the pattern flags as exploits.
Consider pattern $\apattern = (\askeleton, \apred)$ with template $\askeleton = \anop^*_1\cdots\anop^*_k$,
and an execution $\atrace = \anassignment_0 \cdot \anassignment_1 \cdots \anassignment_n$.
Let the sequence of opcodes in $\atrace$ be $\anop_1 \cdots \anop_n$, i.e.,
the transition from $\anassignment_i$ to $\anassignment_{i+1}$ is performed by executing
an instruction with opcode $\anop_{i+1}$.
The pattern $\apattern$ embeds in execution $\pi$ at a subsequence given by indices $(i_1 < \cdots < i_k)$ 
if the corresponding opcodes match the template $\askeleton$: 
$\anop_{i_j} = \anop^*_{j} \text{ for } j \in [1\cdots k]$,
and the execution $\pi$ satisfies the constraint $\apred$.
We denote the fact that $\apattern$ embeds at indices $(i_1, \cdots, i_k)$ in trace $\atrace$
as $\atrace \models_{(i_1, \cdots, i_k)} ~(\askeleton, \apred)$.
Execution $\atrace$ embeds $(\askeleton, \apred)$ if there is a matching subsequence:
$$\atrace \models (\askeleton, \apred) \overset{\Delta}{=} \exists i_1, \cdots, i_k.~~ \atrace \models_{(i_1, \cdots, i_k)} (\askeleton, \apred)$$

Attack pattern $\apattern = (\askeleton, \apred)$ matches 
instruction sequence $\aprogram$ if there is some execution of $\aprogram$ 
that embeds it:
\begin{equation*}
    \aprogram \models \apattern \overset{\Delta}{=} \exists \anassignment \in \initialpred.~ \transrelsem{\aprogram}{\anassignment} \models \apattern
\end{equation*}
Fig. \ref{fig:patternmatch} provides examples of patterns matching instructions.

\newcommand{\setpatterns}{\mathsf{P}}

\subsection{Non-interference Violation Skeleton}
\label{subsec:largeskel}

Given a hyperproperty, we aim to generate a set of patterns such that any 
hyperproperty violation is detected by atleast one of the patterns in this 
set.
However, as illustrated in Example \ref{ex:largeskel} the fact that patterns 
have fixed length is a fundamental limitation in the 
violations they can detect.

\begin{figure}[h]
    \centering
    \includegraphics[scale=0.2]{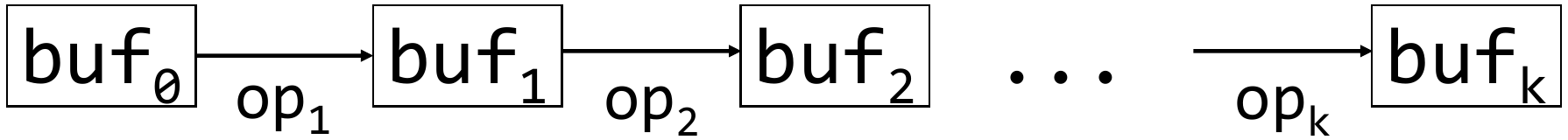}
    \caption{Buffer chain in $\platsynth(k)$ with operations.}
    \label{fig:largeskel}
\end{figure}

\begin{example}[Large skeletons: the $\platsynth(k)$ platform]
    Consider the pedagogical example microarchitecture illustrated in Fig. \ref{fig:largeskel} 
    with $k+1$ state variables (e.g., buffers): $\abuf_0, \cdots, \abuf_k \in \setvars$,
    and $k$ corresponding operations $\anop_1, \anop_2, \cdots, \anop_k \in \setops$.
    Operation $\anop_i$ moves data from $\abuf_{i-1}$ to $\abuf_i$.
    Now, consider a non-interference property with 
    $\setvars_{\pubinp} = \{\abuf_1, \cdots, \abuf_k\}$, and $\setvars_{\obsout} = \{\abuf_k\}$.
    That is we want to identify whether the secret input $\abuf_0$ affects the observable 
    output $\abuf_k$.
    While the operation sequence $\anop_1 \cdots \anop_k$ violates this property
    (it moves data from $\abuf_0$ to $\abuf_k$),
    any sequence of length $k-1$ or less does not.
    \label{ex:largeskel}
\end{example}

\newcommand{\dep}{\mathsf{dep}}
\newcommand{\lwriter}{\mathsf{lw}}
\newcommand{\idep}{\mathsf{idep}}

Given the possibility of large non-interference violations
that would be greater than the size of \textit{any} fixed set of 
patterns, we qualify our problem statement to detecting those
violations with small \textit{skeletons}, as we now define.

Consider an instruction sequence $\aprogram = \aninstr_0 ~ \aninstr_1 ~ \cdots$,
and a corresponding execution trace 
$\atrace = \transrelsem{\aprogram}{\anassignment}$.
For a trace index $i$, we denote the variables that 
$\aninstr_i$ depends on as $\dep_\atrace(i) \subseteq \setvars$.
We denote the last writer of a variable $v\in\setvars$ at index $i$, denoted as 
$\lwriter_\atrace(v, i) \in \{\aninstr_0, \cdots, \aninstr_{i-1}\}$
as the last instruction that writes to $v$ before $\aninstr_i$.
Finally, the set of all (instruction) dependencies of an 
instruction $\aninstr_i$ is the union of the last writers 
of all variables it depends on:
$\idep_\atrace(i) = \bigcup_{v \in \dep_\atrace(i)} \lwriter_\atrace(v, i)$.
For a trace $\atrace$, we say that $i_1, \cdots, i_k$ 
is a subsequence that is \textit{closed under dependencies} if 
for all $j \in [1\cdots k]$,
$\idep(i_j) \in \{i_1, \cdots, i_{k}\}$.

\begin{definition}[Skeleton]
    Suppose the sequence of instructions $\aninstr_1 \cdots \aninstr_n$
    violates a non-interference property $\mathsf{NI}$. 
    Then, $\aprogram$ has an $\mathsf{NI}$-violation skeleton of size $k$,
    denoted as $\aprogram \not\models_k \mathsf{NI}$ if 
    there exist traces $\atrace_1, \atrace_2$ with subsequences of
    length $k$ which are closed under dependencies 
    and also violate $\mathsf{NI}$.
\end{definition}

\subsection{Formal Problem Statement}
\label{subsec:problemstmt}

Formally, given
(a) a platform model (with state $\setvars$ and semantics $\transrel$), and
(b) a non-interference security specification $\mathsf{NI}$($\initialpred$, $\setvars_{\pubinp}$, and $\setvars_{\obsout}$),
and (c) a given skeleton size $k$,
we generate a set of patterns $\setpatterns$ such that 
any instruction sequence $\aprogram$ that violates 
the non-interference property with a skeleton of size $k$ 
is detected by one of the patterns in $\setpatterns$:
\begin{equation}
    \forall \aprogram.~~ \aprogram \not\models_k \mathsf{NI}(\initialpred,\setvars_{\pubinp},\setvars_{\obsout}) 
    \implies \exists \apattern\in \setpatterns.~  \aprogram \models \apattern
    \label{eq:completeness}
\end{equation}

We refer to Eq. \ref{eq:completeness} as the $k$-completeness property.
\section{The SemPat Approach}
\label{sec:approach}

\newcommand{\gentemp}{\textbf{GenerateTemplates}}
\newcommand{\conspec}{\textbf{ConstraintSpecialize}}

\begin{figure*}[h]
    \centering
    \includegraphics[width=0.9\textwidth]{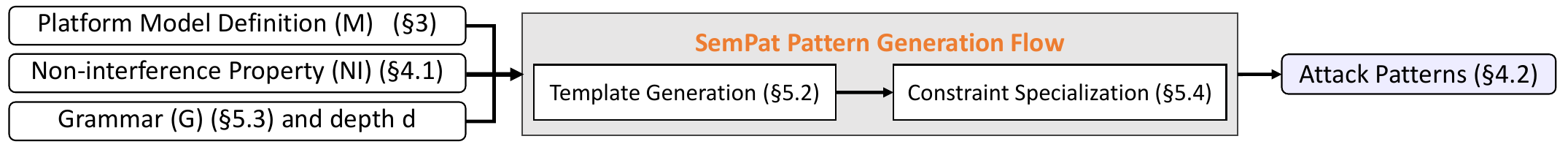}
    \caption{The SemPat approach.}
    \label{fig:sempatpipeline}
\end{figure*}

In this section we discuss our approach to automatically generate 
attack patterns satisfying $k$-completeness (Eq. \ref{eq:completeness}).
We provide an overview in \S\ref{subsec:gen-overview}, 
followed by details 
in \S\ref{subsec:skelexplore} and \S\ref{subsec:grammar-search}.

\begin{figure}
    \centering
\begin{lstlisting}[style=pseudostyle,mathescape=True]
type rentry_t = record { op1: word_t, op2: word_t, result: word_t };
var reuse_buf : [instr_ind_t]rentry_t;

operation mulop (rd, rs1, rs2) {
  op1 = regfile[rs1]; op2 = regfile[rs2];
  // Check for reuse
  if ($\exists$ i. reuse_buf[i].op1 == op1 &&
      reuse_buf[i].op2 == op2) {
    result = reuse_buf[i].result; // Reuse result
  } else {  // Otherwise invoke multiplier unit
    result = multiplier(op1, op2);
    mulcount = mulcount + 1;
  }
  // Replace reuse_buf entry (at index ind_)
  reuse_buf[ind_] = rentry_t {op1, op2, result};
  regfile[rd] = result;
}
\end{lstlisting}
    \caption{$\platcompreuse$: Fragment of the computation reuse platform model.}
    \label{fig:platcolisting}
\end{figure}

\subsection{Pattern Generation Overview}
\label{subsec:gen-overview}

\subsubsection*{Operation}

Figure \ref{fig:sempatpipeline} provides a high-level 
overview of our approach.
We take in a transition system-based platform model 
(\S\ref{sec:progmodel}),
a non-interference specification (\S\ref{subsec:securitydef}), 
a constraint grammar $\agrammar$ and depth $\adepth$. 
We generate a set of patterns $\setpatterns$ up to depth $\adepth$, 
with constraints sourced from $\agrammar$ (discussed in \S\ref{subsec:patgrammar}).
Our approach has two components: template generation (\S\ref{subsec:skelexplore}) and 
grammar-based specialization (\S\ref{subsec:grammar-search}).
We explain these elements using a running example.

\subsubsection*{Running Example: Computation Reuse}
We use the running example of the generation of the
$\fpop - \ldop - \fpop$ pattern from Fig. \ref{fig:patternexample}.
We generate this pattern based on the computation-reuse 
platform model ($\platcompreuse$), an excerpt of which is provided in 
Fig. \ref{fig:platcolisting}.
The $\platcompreuse$ microarchitecture includes a reuse
buffer ($\rbuf$) that stores the operands and results of previous 
multiplication (\texttt{mul}) instructions.
Future \texttt{mul} instructions matching operands of 
previous instructions, can reuse results from the buffer 
instead of reinvoking the multiplier.
We encode the security property that secret data from the 
memory should not affect the count of multiplier invocations
as the property
$\mathsf{NI}_\platcompreuse(\mathsf{init}_\rbuf, \texttt{mem}, \texttt{mulcount})$
where $\mathsf{init}_\rbuf$ constrains 
all entries in the $\rbuf$ to be initially invalid and
$\texttt{mulcount} \in \setvars$ counts multiplier invocations.

\subsection{Template Generation}
\label{subsec:skelexplore}

The first phase of our approach uses an overapproximate analysis
to generate templates up to the user-specified depth $\adepth$.
This is performed by the \gentemp{} procedure. We discuss \gentemp{} intuitively and illustrate 
it with an example, while defering its pseudocode to the appendix.
\gentemp{} scans over all templates starting with size 1
(single operations) up to size $\adepth$.
For each template, we first perform (overapproximate) taint analysis 
to check whether the template propagates taint from the secret inputs 
($\setvars_\secinp$) to the public outputs ($\obsvars$).
If it does not, then the template will also not violate the 
target non-interference property, and we can skip over it.
If it does propagate taint from $\setvars_\secinp$ to $\obsvars$, 
we check if the template violates the target non-interference property
\textit{semantically}.
We do this by reducing non-interference to a safety query 
(e.g., \cite{Rushby1992NoninterferenceTT,Clarkson2008Hyperproperties}),
and solving this query using SMT-based model checking (\cite{SymbolicMC,Barrett2009SatisfiabilityMT}).

\begin{example}[Template generation for $\fpop - \ldop - \fpop$]
    Consider invoking \gentemp{} for the 
    $\platcompreuse$ platform (Fig. \ref{fig:platcolisting}) with
    $\mathsf{NI}_\platcompreuse$ as the non-interference property
    and a search depth of 3.
    \gentemp{} first considers single operation templates.
    However, none of them propagate taint from 
    \texttt{mem} to \texttt{mulcount}.
    Subsequently, \gentemp{} finds 
    the two-operation template $\ldop - \fpop$ does 
    propagate taint from \texttt{mem} to \texttt{mulcount}.
    However, this template does not (semantically) violate non-interference 
    as the $\rbuf$ is initially empty, and the multiplier
    must be invoked in all executions.
    Eventually, \gentemp{} considers the 
    size 3 template $\fpop - \ldop - \fpop$. 
    This template both, propagates taint and semantically violates $\mathsf{NI}_\platcompreuse$.
    This is because the first $\fpop$ `primes' the reuse buffer, 
    and the second $\fpop$ uses the primed buffer entry.
    This will lead to two executions with different 
    $\texttt{mulcount}$ values, i.e., a non-interference violation.
\end{example}

\gentemp{} generates all operation sequences (up to length $d$) 
that violate the non-interference property:
\begin{lemma}
    \label{lem:templategen}
    For $k \leq d$, if $\aninstr_1 \cdots \aninstr_k \not\models \mathsf{NI}$
    where instruction $\aninstr_i = \anop_i(\anopermap_i)$, then, $\anop_1 \cdots \anop_k \in \gentemp{}(M, \mathsf{NI}, d)$.
\end{lemma}

\setcounter{algocf}{1}

\begin{algorithm}
    \small
\SetKwFunction{GrammarHelper}{ConsHelper}
\SetKwProg{Fn}{Function}{:}{}
\DontPrintSemicolon
    \KwInput{Semantic platform definition $M$, non-interference property $\mathsf{NI}$, pattern template $\askeleton$, and a grammar $G$}
    \KwOutput{A set of patterns}
    \KwData{acc: an accumulated set of patterns}
    
    \Fn{\GrammarHelper{$\apred, i, L$}}{
        \tcc{Exhausted all atomic predicates?}
        \lIf{$i > |L|$}
        {
            acc.append($(\askeleton, \apred)$)
        }
        \tcc{Does adding $\neg L[i]$ eliminate all violations?}
        \lElseIf{
            $\forall \aninstr_1, \cdots, \aninstr_{|\askeleton|}.~  
            \aninstr_1 \cdots\aninstr_{|\askeleton|} \models (\askeleton, \apred \land \neg L[i]) \implies
            \aninstr_1, \cdots, \aninstr_{|\askeleton|} \models \mathsf{NI}$ \;
        }
        { \GrammarHelper{$\apred \land L[i]$, $i+1, L$} \tcc*[f]{add $L[i]$}}
        \lElse(\tcc*[f]{skip over $L[i]$}){
            \GrammarHelper{$\apred, i+1, L$};
        }
    }
    $L$ = ApplyPredicates($\askeleton, G$) \tcc*[r]{Create all atoms in $L$}
    \GrammarHelper{$\texttt{true}, 0, L$} \tcc*[r]{Counterfact. addition}
    \Return acc
\caption{\conspec{}($M, \mathsf{NI}, \askeleton, G$)}
\label{alg:gram-search}
\end{algorithm}

\begin{table*}
    \centering
\resizebox{\textwidth}{!}{
    \begin{tabular}{l|l|l}
    \hline
    Predicate Atom & Meaning & Encoding (assuming RISC-V ISA) \\ \hline
    $\datadep(\aninstr_1, \aninstr_2)$ & Data operand of $\aninstr_2$ depends on result of $\aninstr_1$ & Last write to $\aninstr_2.\rsone$ or $\aninstr_2.\rstwo$ is by $\aninstr_1$ \\
    $\addrdep(\aninstr_1, \aninstr_2)$ & Address operand of $\aninstr_2$ depends on result of $\aninstr_1$ & Last write to $\aninstr_2.\rsone$ is by $\aninstr_1$ \\
    $\sameaddr(\aninstr_1, \aninstr_2)$ & Address operands of $\aninstr_1$ and $\aninstr_2$ are the same & $\aninstr_2.\addr = \aninstr_1.\addr$ (for memory operations) \\
    $\diffaddr(\aninstr_1, \aninstr_2)$ & Address operands of $\aninstr_1$ and $\aninstr_2$ are different & $\aninstr_2.\addr \neq \aninstr_1.\addr$ (for memory operations) \\ \hline
    $\srcdata_\areg(\aninstr_1)$ & Data operand is read from register $\areg$ & e.g., $\aninstr_1.\rsone = \areg$, $\aninstr_1.\rstwo = \areg$ (depends on opcode) \\
    $\srcaddr_\areg(\aninstr_1)$ & Address operand is read from register $\areg$ & $\aninstr_1.\rsone = \areg$ (for memory operations) \\
    $\destreg_\areg(\aninstr_1)$ & Result is written to a register $\areg$ & $\aninstr_1.\rd = \areg$ \\ \hline
    $\speculates(\aninstr)$ & Instruction $\aninstr$ initiates speculation & $\aninstr$ sets the $\specvar$ variable \\ \hline
    $\lowoperands(\aninstr_1)$ & Operands of $\aninstr_1$ is independent of $\setvars_\secinp$ ($= \setvars\setminus\srcvars$) & e.g., $\sigma_1(\regfile[\rsone]) = \sigma_2(\regfile[\rsone])$ (depends on opcode) \\
    $\lowresult(\aninstr_1)$ & Result of $\aninstr_1$ is independent of $\setvars_\secinp$ ($= \setvars\setminus\srcvars$) & $\sigma_1(\regfile[\rd]) = \sigma_2(\regfile[\rd])$ \\ \hline
    \end{tabular}
}
    \caption{Pattern predicate grammar: we generate patterns with constraints as conjunctions of these predicate atoms.}
    \label{tab:facetgrammar}
\end{table*}

\subsection{Conjunction-based Pattern Grammar}

\label{subsec:patgrammar}

While templates alone are too overapproximate to be useful, augmenting (specializing) 
them with constraints ($\apred$) leads to more precise patterns with fewer false positives.
Specialization is performed using the user provided grammar $G$ that identifies the 
space of these constraints.

\subsubsection{Conjunction of Predicate Atoms}

\label{subsubsec:conjgrammar}

We consider constraints which are conjunctions of 
atoms.\footnote{In formal logic, an atom (atomic formula) is a single (indivisible) logical proposition.}
That is, $\apred$ has form: $\apred = \bigwedge_{i} \afacet_i$,
where each $\afacet_i$ is an atom.
These atoms ($\afacet_i$) are generated by applying predicates from a grammar $G$,
such as the one in Table \ref{tab:facetgrammar}, which we use as the default.
Each predicate from the grammar is applied to some
number of instructions from the pattern, as indicated
by its arity
(e.g., $\datadep$ is an arity-2 (binary) predicate).
To identify the instructions a predicate is applied to,
we distinguish apart identical opcodes
using their position
(e.g., the $\datadep$ predicate is applied to ($0:\fpop$), ($2:\fpop$) in Example \ref{ex:pattern}).

\subsubsection{Precision vs. Robustness Tradeoff}
\label{subsubsec:precisionrobustness}

A grammar that only allows high-level (architectural) predicates (e.g., Tab. \ref{tab:facetgrammar})
leads to patterns which are less sensitive to microarchitectural implementation
details, but have more false positives.
Conversely, a grammar that exposes low-level microarchitectural constraints leads 
to more precise patterns (with fewer false positives). 
However, these patterns are then specific to the platform microarchitecture.
Thus, the pattern grammar exposes a tradeoff between the
robustness and expressivity/precision of generated patterns.
While our specialization technique (\S\ref{subsec:grammar-search}) 
requires a conjunction-based grammar (\S\ref{subsubsec:conjgrammar}),
we are not fundamentally limited to the predicates from Table \ref{tab:facetgrammar}.
We explore this further in \S\ref{subsec:rq4} where we augment 
the default grammar with additional predicates to improve precision.

\subsection{Template Specialization with Predicates}
\label{subsec:grammar-search}

The goal of template specialization is making patterns as precise as possible
using the pattern grammar (\S\ref{subsec:patgrammar}),
while ensuring that they do not miss any
violating executions (as required by Eq. \ref{eq:completeness}). 
Specialization is performed by invoking the \conspec{} procedure (Alg. \ref{alg:gram-search}) 
on every pattern template generated by \gentemp{}.
At a high level, starting with the $\texttt{true}$ constraint (line 0),
\conspec{} continues adding predicate atoms to the constraint. 

\subsubsection{Candidate predicate atoms}
\label{subsubsec:gramatoms}

As introduced in \S\ref{subsec:patgrammar}, the pattern 
constraint is a conjunction of predicate atoms from grammar $G$. 
We apply each predicate to operations from the template 
to get a set of atoms.
For example, the binary $\datadep$ predicate with 
the $(0: \fpop) - (1: \ldop) - (2: \fpop)$ template
results in three atoms: $\datadep(0: \fpop, 1: \ldop)$, 
$\datadep(1: \ldop, 2: \fpop)$ or $\datadep(0: \fpop, 2: \fpop)$
(we ignore backwards dependencies).
\conspec{} first collects all such atoms in a list $L$
using the \text{ApplyPredicates} helper function (line 6).

\subsubsection{Counterfactual-based atom addition}
\label{subsubsec:counterfactual}

Adding an atom strengthens the pattern constraint, resulting 
in it capturing fewer executions.
To ensure that the generated pattern does not miss 
any non-interference violations (and thereby violate 
Eq. \ref{eq:completeness}) we use \textit{counterfactual atom addition}.
Counterfactual addition adds an atom only if adding the 
negation of the atom leads to only non-violating executions.
Intuitively, if the negation leads to only non-violating executions,
then adding the atom preserves all violating executions:
\begin{observation}
    \label{obs:counterfactual}
    Consider a pattern $(\askeleton, \apred)$, where $|\askeleton| = k$, 
    and an atom $\afacet$. We have, for all $\aprogram = \aninstr_1 \cdots \aninstr_k$:
\begin{align*}
    &\big(\aprogram \models (\askeleton, \apred \land \neg \afacet) \implies \aprogram \models \mathsf{NI}\big) \implies \\
    &\quad\quad \big((\aprogram \models (\askeleton, \apred) \land \aprogram \not\models \mathsf{NI}) \implies \aprogram \models (\askeleton, \apred \land \afacet)    \big)
\end{align*}
\end{observation}
For each atom in $L$, we check (Alg. \ref{alg:gram-search} line 3) 
if it satisfies the counterfactual addition condition,
specializing the pattern (line 4) if so.
If not, we skip over it (line 5) and move to the next atom in $L$.

\begin{example}
    For the $(0: \fpop) - (1: \ldop) - (2: \fpop)$ template,
    \conspec{} adds the $\afacet = \datadep(1: \ldop, 2: \fpop)$ 
    atom to the constraint, 
    since if $(2: \fpop)$ does not depend on $(1: \ldop)$,
    then its operands are not secret, and the 
    non-interference property is not violated.
    In a further iteration, \conspec{} adds the 
    atom $\highresult(1: \ldop)$ which says that the loaded value 
    is secret-dependent. Once again, the negation
    of this would lead to a non-violating execution.
    This gives us the final pattern with
    $w = (0: \fpop) - (1: \ldop) - (2: \fpop)$ and
    $\apredicate = \datadep(1: \ldop, 2: \fpop) \land \highresult(1: \ldop)$.
\end{example}

\subsubsection{Multiple counterfactuals and branching}

Counterfactual addition relies on a strong condition 
which may not hold for single atoms.
In such cases, we consider multiple counterfactual atoms:

\begin{observation}[Multiple counterfactuals]
    \label{obs:multicounterfactual}
    For pattern $(\askeleton, \apred$), 
    atoms $\{\afacet_i\}_i$, and for any $\aprogram = \aninstr_1 \cdots \aninstr_{|\askeleton|}$:
    \begin{align*}
        &\big(\aprogram \models (\askeleton, \apred \land 
        \bigwedge_i \neg \afacet_i) \implies \aprogram \models \mathsf{NI}\big) \implies \\
        &\quad\quad \big((\aprogram \models (\askeleton, \apred) \land \aprogram \not\models \mathsf{NI})\implies {\boldsymbol \bigvee_i} (\aprogram \models (\askeleton, \apred \land \afacet_i))\big)
    \end{align*}
\end{observation}

Considering multiple counterfactuals $\{\afacet_i\}$ 
results in a disjunctive branching (bolded $\bigvee_i$) over atoms.
\conspec{} recursively invokes \texttt{ConsHelper} on $(\askeleton, \apred \land \afacet_i)$
for each $i$.
This is sound as the collection of patterns 
together continue to cover all violating executions.
We do not include multi-counterfactuals in Alg. \ref{alg:gram-search} pseudocode for brevity;
we provide full pseudocode in the appendix.

\begin{example}[Multiple counterfactuals for $\platsynth$]
\label{ex:platsynthmulticounterfactual}
Consider the platform $\platsynth(k)$ from Ex. \ref{ex:largeskel},
with parameter $k = 2$ and grammar depth $d = 3$.
\gentemp{} returns the template 
$(0: \anop_0) - (1: \anop_0) - (2: \anop_1)$
for which neither atom $\datadep(0: \anop_0, 2: \anop_1)$
nor $\datadep(1: \anop_0, 2: \anop_1)$ 
can be added alone.
This is because, even if one is negated, the other might be true, 
leading to a non-interference violation.
However, on adding both negations, the violation is blocked.
Our approach then will branch and specialize each case further.
\end{example}

By Lem. \ref{lem:templategen} and 
Obs. \ref{obs:counterfactual},\ref{obs:multicounterfactual}, 
we have $k$-completeness of our approach up to skeleton size $k=d$.
This is formalized through the following theorem and
proved by induction on the number of specialization iterations.
We provide a proof in the appendix.

\begin{theorem}
    \label{thm:kcomplete}
    Let $W = \gentemp{}(M, \mathsf{NI}, d)$ be the templates for depth $d$, and let 
    $\setpatterns_i = \conspec{}(M, \mathsf{NI}, \askeleton_i, G)$ be the specialized patterns 
    for each $\askeleton_i \in W$. 
    Then, for all instruction sequences $\aprogram = \aninstr_1 \cdots \aninstr_k$, {with $k \leq d$}, we have:
    \begin{equation*}
        \aprogram \not\models \mathsf{NI} \implies \exists \apattern \in \bigcup_i \setpatterns_i. ~ \aprogram \models \apattern
    \end{equation*}
\end{theorem}

\section{Evaluation Methodology}
\label{sec:methodology}

\subsection{Tool Prototype}

We implement our approach in a prototype tool $\toolname$.
$\toolname$ allows us to generate patterns, 
and analyze program binaries using these patterns.
For pattern generation, $\toolname$ allows the user to specify 
the platform description (\S\ref{sec:progmodel}), 
and the non-interference specification (\S\ref{subsec:securitydef}).
By default pattern generation uses the predicate grammar from Tab. \ref{tab:facetgrammar}.
However, we also allow the user specify custom predicates.
Based on these inputs the tool generates a set of patterns as 
described in \S\ref{sec:approach}.
The resulting patterns can be inspected by the user, and used for binary analysis.
Analysis can either be performed using patterns or 
the hyperproperty specification.
We discuss the details of this analysis in \S\ref{subsec:patternmatching}.

$\toolname$ uses the UCLID5 \cite{uclid5} verification engine as the 
backend model checker for both, the non-interference checks involved 
in pattern generation as well as binary analysis.
UCLID5 internally compiles model checking queries into SMT 
\cite{BarFT-SMTLIB, Barrett2009SatisfiabilityMT} queries, and
invokes an SMT solver (e.g. Z3 \cite{Moura2008Z3AE}, CVC5 
\cite{Barbosa2022cvc5AV}) on them.

\subsection{Platform Designs}

While users can specify their own platform models,
we describe the models considered in our experimentation.

\subsubsection{$\platcompreuse$: Computation Reuse}
\label{subsec:platco}

The computation reuse platform ($\platcompreuse$, Fig. \ref{fig:platcolisting}) 
was introduced in \S\ref{subsec:gen-overview}.
It is based on microarchitectural optimizations that enable dynamic 
reuse of results of previous high-latency computations.
Our model includes a \textit{reuse buffer} (\pseudocb{reuse_buf}), 
which records operands and results of multiply operations, 
which are then reused by future operations with matching operands. 
Dynamic (in-hardware) optimization schemes have been proposed in the literature 
(e.g., \cite{Mutlu2005OnRT, Sodani1997DynamicIR, Charles2002ImprovingIW})
and also have been hypothesized to be vulnerable to computational side-channel attacks \cite{Vicarte2021OpeningPB}.
Our model is based on the scheme proposed in \cite{Sodani1997DynamicIR}.

\subsubsection{$\platss$: Store Optimization}
\label{subsec:platss}

The store optimization platform ($\platss$) models 
microarchitectural optimizations related to the store unit 
in the memory hierarchy.
$\platss$ models the silent stores \cite{Lepak2000OnTV,Lepak2001SilentSA} 
and store-to-load forwarding optimizations.
The former squashes (makes \textit{silent}) stores that would rewrite 
the same value to the memory. 
Store-to-load forwarding serves store values to 
subsequent loads on the same address.
While these optimizations have several implementation variants 
(e.g. \cite{Sha2005ScalableSF,Sha2006NoSQSC,Subramaniam2006FireandForgetLS}),
we base our model on the relatively simple \textit{LSQ cache} 
(load-store queue cache) 
\cite{Lepak2000SilentSF} which allows combining these optimizations.

We illustrate the LSQ cache through an excerpt of our model in Fig. \ref{fig:platsslisting}.
The LSQ cache is a FIFO buffer that records load and store requests.
On a \textit{load}, the buffer is checked for store requests to the same address.
If there exists such a request (with no subsequent stores to the same address),
then the load can be sourced with this store.
On a \textit{store} too, the buffer is checked for a
request, and the store is squashed (silenced) 
if the request payload matches the value written by the store.
If a valid entry is not found, the new load/store performs 
a full memory request as usual and the LSQ cache is updated.

\begin{figure}
  \centering
\begin{lstlisting}[style=pseudostyle]
type lsqc_entry_t = record { valid: boolean, 
  is_load: boolean, addr: word_t, data: word_t }
var lsqc : [lsqc_index_t]lsqc_entry_t;
var lptr : lsqc_index_t; // Pointer into LSQC
var lscount : int; // Number of memory accesses

operation store (rs2, rs1, imm) {
  addr = addrgen(regfile[rs1], imm);
  data = regfile[rs2];
  if (lsqc[lptr].valid && lsqc[lptr].addr == addr 
      && lsqc[lptr].data == data) {
      // Store squashing (silent store)
  } else {            // Perform memory access
      mem[addr] = data; lscount += 1;
      lsqc = ...      // Update the lsqc
  }
}
operation load (rd, rs1, imm) {
  addr = addrgen(regfile[rs1], imm)
  if (lsqc[lptr].addr == addr && lsqc[lptr].valid){
    data = lsqc[lptr].data; // Data forwarding
  } else {              // Perform memory access
    data = mem[data]; lscount += 1; 
    lsqc = ...          // Update the LSQC
  }
  regfile[rd] = data;
}
\end{lstlisting}
  \caption{$\platss$: Excerpt of the abstract platform modelling load-store optimizations in the microarchitecture.}
  \label{fig:platsslisting}
  \vspace*{-0.2cm}
\end{figure}

\subsubsection{Speculation primitives}
\label{subsubsec:specprimitives}

In addition to the above optimizations, our 
consider platform models have 
branch and store-to-load forwarding speculative features.
We model speculation non-deterministically, abstracting away
from the speculative choice mechanism (e.g. branch predictor).
An instruction that speculates moves the platform into 
speculative mode by setting the $\specvar$ variable 
(described in \S\ref{sec:progmodel}).
Our current implementation only supports a single speculative frame.
For example, we do not allow speculating on loads
within a branch speculation context.
However, this is not a fundamental limitation of our approach.

\subsection{Binary Analysis}
\label{subsec:patternmatching}

We first disassemble the binary to obtain the control flow graph (CFG).
We unroll the control flow graph up to a fixed depth,
starting from a particular entry point.
We instrument the unrolling at the branch instructions with 
assumptions corresponding to the branch condition,
allowing bypassing the condition when speculating
(i.e., if $\specvar$ is set).
For example ``\texttt{if (x == 0) stmt}'' is instrumented as
``\texttt{assume(x == 0 || $\specvar$); stmt}''.
Each instrumented unrolling is analyzed using a 
pattern/hyperproperty based check.

\textbf{Pattern-based analysis.} Given an unrolling, the pattern-based 
analysis implements the pattern check described in \S\ref{subsubsec:attpat}.
That is, we first identify all subsequences in that unrolling that match
the pattern template. 
For each such subsequence, we formulate an SMT query that checks
whether the subsequence satisfies the pattern constraint, 
and solve it using the UCLID5 model checker \cite{uclid5}.

\textbf{Hyperproperty-based analysis.} The hyperproperty-based analysis
follows existing approaches (e.g., \cite{Cheang2019AFA,Guarnieri2020SpectectorPD}) 
by encoding the hyperproperty as a safety property over a self-composition 
(\cite{Terauchi2005SecureIF,Clarkson2008Hyperproperties}) 
of the platform model.
The copies in the self-composition execute identical unrolled programs.
Once again we use the UCLID5 model checker to solve the generated safety query.

\section{Experimental Results}
\label{sec:evaluation}

\subsubsection*{Experimental Setup}

We perform experiments on a machine with an
Intel i9-10900X CPU with 64GB RAM
rated at 3.70GHz and running Ubuntu 20.04.
As mentioned in \S\ref{sec:methodology}, we use the
UCLID5 model checker \cite{uclid5} and 
Z3 (version 4.8.7) \cite{Moura2008Z3AE} as the backend SMT solver
for the queries that UCLID5 generates.
We mention timeouts for each experiment in its respective
section.

\subsubsection*{Research Questions}

We aim to answer the following questions through our evaluation:
(a) Can the approach generate new patterns, can these patterns be used
to detect attacks, do they provide performance advantages?
(b) How well does our approach scale with platform complexity,
and the pattern generation search depth?
(c) How does the choice of grammar affect false positives?

\begin{figure}
    \centering
    \includegraphics[scale=0.24]{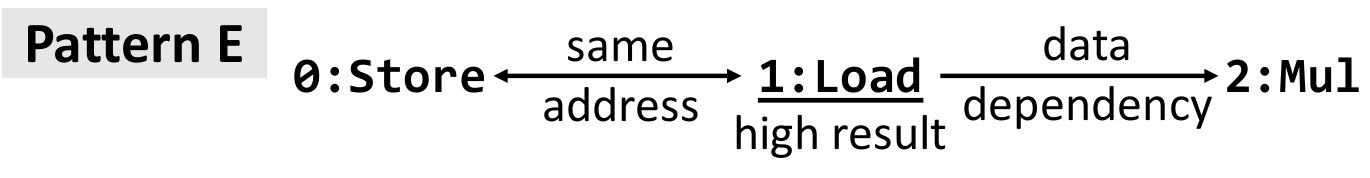}
    \caption{Pattern for store-to-load forwarding. 
    Explanation: The store address matches the load, however,
    due to address computation latency, the load is speculatively
    (incorrectly) served 
    from memory while the store is still in the store buffer and has not propagated to memory.}
    \label{fig:stlpattern}
\end{figure}

\begin{figure}
    \centering
    \includegraphics[scale=0.27]{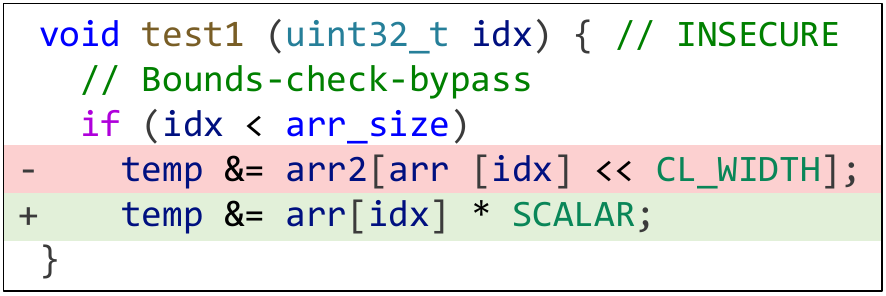}
    \includegraphics[scale=0.27]{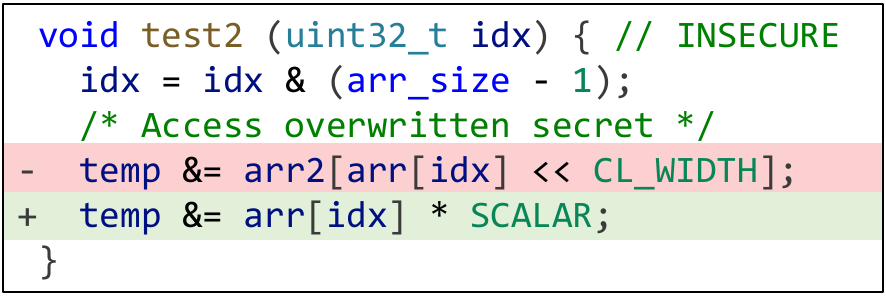}
    \caption{Modification of SpectreV1 (left) and SpectreV4 (right) litmus tests from 
    \cite{Daniel2021HuntingTH} to target a computation-unit side channel (instead of cache-based side channel).} 
    \label{fig:branch-stl-examples}
\end{figure}

\begin{figure*}
    \centering
    \includegraphics[scale=0.4]{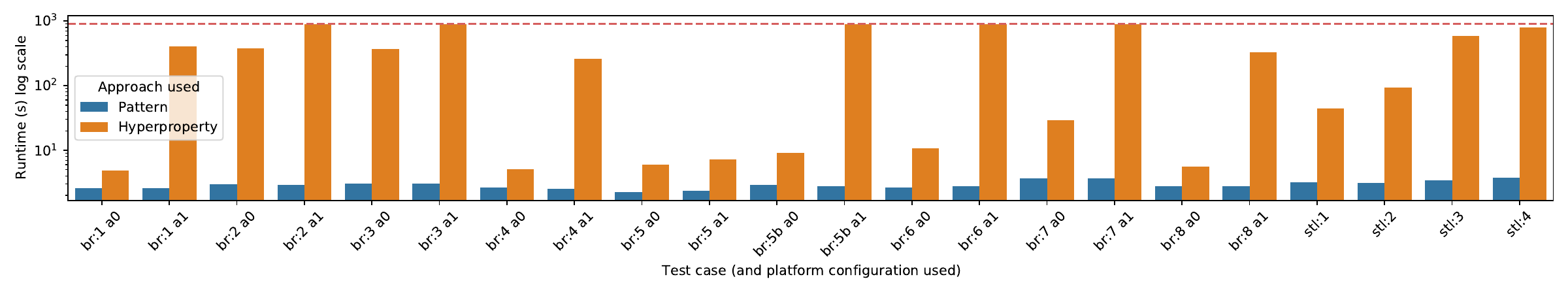}
    \caption{Verification of Branch and STL speculation litmus tests (see Fig. \ref{fig:branch-stl-examples}) with hyperproperties and generated patterns.}
    \label{fig:branch-stl-checktimes}
\end{figure*}

\subsection{RQ1: Can we generate patterns for 
new vulnerability variants?}
\label{subsec:rq1}

To answer \textbf{RQ1}, we considered two extensions to the computation
reuse platform ($\platcompreuse$) introduced in \S\ref{sec:methodology}.
In these extensions we modelled branch speculation 
and store-to-load forwarding respectively (as discussed in \S\ref{subsubsec:specprimitives}).
We then generated patterns for these platforms based on 
the speculative non-interference property $\text{H}_\text{CD}$
(\S\ref{sec:motivation}). 
We used a $\texttt{reuse\_buf}$ of size 4 (i.e., 4 buffer entries), 
and grammar search depth $d = 3$.

\begin{table}[h]
    \centering
    \resizebox*{0.45\textwidth}{!}{
    \begin{tabular}{c|c|c|c|c}
        \hline
        \makecell{Platform} & \makecell{Gen. time  \\ (ww=16)} & \makecell{Gen. time  \\ (ww=20)} & \makecell{Gen. time  \\ (ww=24)} & \makecell{Gen. time  \\ (ww=32)} \\ \hline
        Branch + CR & \multicolumn{3}{c|}{--} & 69s  \\ \hline
        STL + CR & 92s & 334s & \multicolumn{2}{c}{ TO (1 hour)}
    \end{tabular}
    }
    \caption{Generation times for branch speculation and store-to-load forwarding patterns.}
    \label{tab:branchspec-gen-times}
    \vspace*{-0.2cm}
\end{table}

Our approach was able to generate Pattern C and Pattern D from Fig. \ref{fig:programexamples}
for branch speculation with the computation-unit side channel.
For store-to-load forwarding, we obtain the pattern shown in Fig. \ref{fig:stlpattern}.
We provide the time required to generate these patterns in Table \ref{tab:branchspec-gen-times}.
Our STL+CR generation timed out on a particular non-interference check query 
with a machine word width of 32 bits, hence, we used lower word widths.\footnote{
We believe that this is an underlying issue with the solver and are investigating it further.}
The same hyperproperty ($\text{H}_\text{CD}$) was used for both (Branch and STL) cases, 
demonstrating that hyperproperties allow uniform vulnerability specification.
\textit{This demonstrates the ability of our approach to generate new patterns.}

\subsection{RQ2: Can the generated patterns efficiently find exploits
or prove their absence?}
\label{subsec:rq2}

To answer \textbf{RQ2}, we evaluate whether the generated 
patterns provide advantages over hyperproperty-based detection.
For this, we considered modified versions of 
SpectreV1 and SpectreV4 litmus tests from \cite{kocher, Daniel2021HuntingTH}.
In each of these examples, we replaced the secret-dependent load 
with a multiply operation (targetting the computation-based side channel).
We call these examples SpectreV1-CR and SpectreV4-CR respectively,
and illustrate two tests (one for V1 and one for V4) 
in Fig. \ref{fig:branch-stl-examples}.
We consider 9 such tests (8 unsafe, 1 safe) for SpectreV1-CR and 
4 tests (all unsafe) for SpectreV4-CR.

We then evaluated both hyperproperty-based and pattern-based
detection approaches on these tests with a timeout of 15 mins per test.
We provide the results in Fig. \ref{fig:branch-stl-checktimes}.
For each SpectreV1-CR variant test, we considered two platform 
configurations: one with a non-associative cache, (denoted as
\texttt{br:<n> a0}) and one with an associative cache
(these are marked as \texttt{br:<n> a1}).
The pattern-based approach is able to check all test cases 
(unsafe and safe) correctly within ($\sim$5 seconds).
In contrast, the hyperproperty-based approach often takes upwards
of 2 mins, with some tests (e.g., \texttt{br:2}, \texttt{br:3}) even timing out 
(i.e., taking greater than 15 mins).
An interesting observation is that there are some cases
(e.g., \texttt{br:1}, \texttt{br:4}) where the hyperproperty-based check
converges in a reasonable time ($\sim$10s) with the non-associative cache (\texttt{a0}), 
while taking much longer (or times out) with the associative cache (\texttt{a1}).
The pattern-based approach is resilient to these 
differences since it operates over architectural state.
This experiment demonstrates how, by abstracting away complex platform 
microarchitecture, \textit{patterns enable security verification that scales
much better that hyperproperties}.

\subsection{RQ3: How well does pattern generation scale?}
\label{subsec:rq3}

We explore the scalability of our approach with the complexity of the platform
and the depth of the pattern grammar.

\begin{figure}
    \centering
    \includegraphics[scale=0.35]{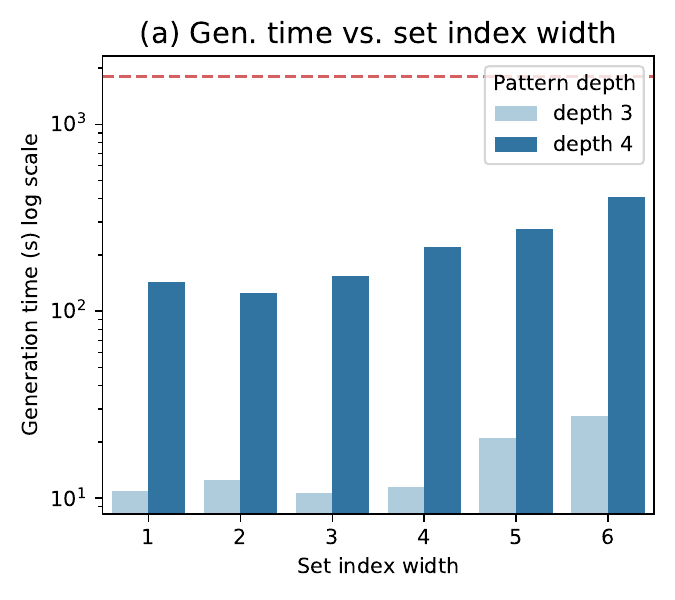}
    \includegraphics[scale=0.35]{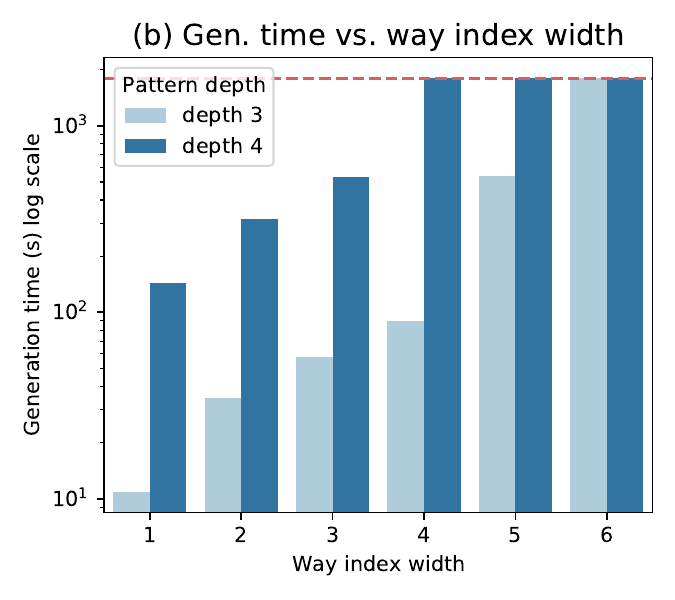}
    \includegraphics[scale=0.35]{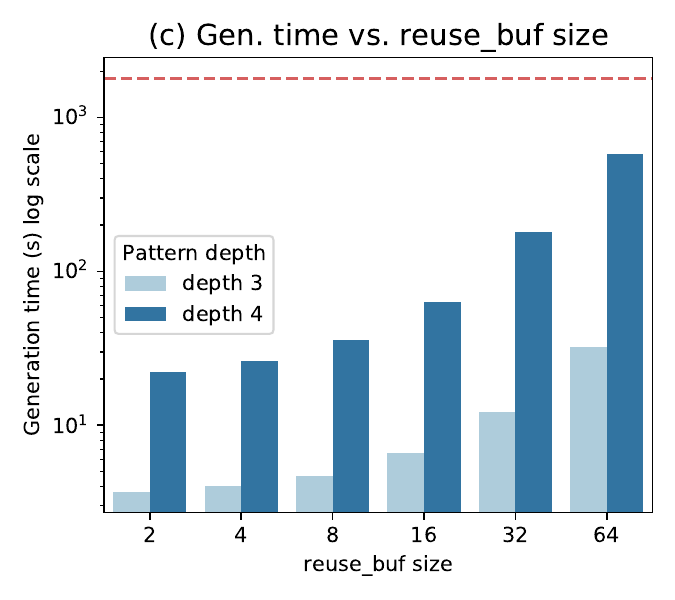}
    \includegraphics[scale=0.35]{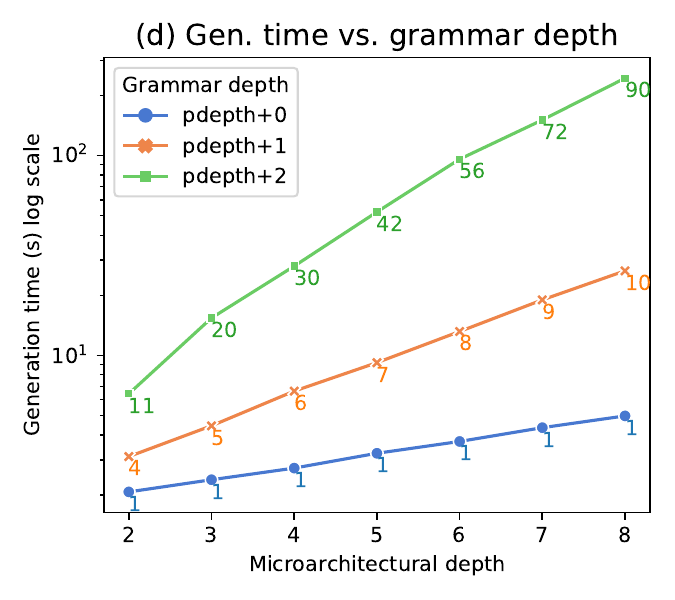}
    \caption{Scalability of pattern generation vs. platform complexity and grammar depth: (a, b) varying set and way index width of the LSQ Cache (Fig. \ref{fig:platsslisting}),
    (c) varying size of the reuse buffer ($\rbuf$ in Fig. \ref{fig:platcolisting}),
    and (d) generation time and number of generated patterns for $\platsynth$ 
    with varying microarchitectural ($\pdep$) and grammar ($\gdep$) depths (Ex. \ref{ex:largeskel}).}
    \label{fig:synthetic-generation}
    \vspace*{-0.3cm}
\end{figure}

\subsubsection{Scalability with model parameters}
\label{subsubsec:paramscaling}

To evaluate this, we experiment on the $\platss$ (\S\ref{subsec:platss}) and 
$\platcompreuse$ (\S\ref{subsec:platco}) platforms.
For $\platss$, we specify a non-interference hyperproperty with $\setvars_\secinp = \{\texttt{mem}\}$,
as the secret input,
and $\obsvars = \{\pseudocb{lscount}\}$
(number of memory operations propagating to memory) 
as the public output.
For $\platcompreuse$, we specify a non-interference property with 
$\setvars_\secinp = \{\texttt{mem}\}$, and $\obsvars = \{\pseudocb{mulcount}\}$,
i.e., number of multiplier invocations as the public output.
We vary size parameters of the LSQ cache (\S\ref{subsec:platss}) and 
reuse buffer (\S\ref{subsec:platco}) components respectively.
For the LSQ cache we sweep through the set index width (same associativity, larger
number of sets), and the way index width (same sets, larger associativity) 
and present results in Fig. \ref{fig:synthetic-generation} (a, b).
For the reuse buffer, we sweep through different buffer sizes
and present the results in Fig. \ref{fig:synthetic-generation} (c).
We perform generation with depths 3 and 4, and a timeout of 30 mins.

\textbf{Observations:}
In Fig. \ref{fig:synthetic-generation}(a, b), we
observe that while generation time grows with both set and way width, 
the increase with way index width is more rapid.
Our hypothesis for this is that a larger number of sets
increases the width of indexing bitvectors,
while larger associativity increases the number case-splits
(conditionals) since ways are iterated over.
The former increase is slower due to the use of word-level reasoning
in SMT solvers, while the latter has an explicit if-then-else encoding,
and hence a larger formula.
We observe a similar increase for the reuse buffer.
While the runtime increases exponentially, 
our approach is reasonably efficient and is able to generate 
patterns on the order of minutes for realistic sizes.

\subsubsection{Scalability with search depth}
\label{subsubsec:depthscaling}

In order to investigate the scalability of our approach with the depth of the grammar,
we perform experimentation with the synthetic platform $\platsynth(k)$
introduced in Ex. \ref{ex:largeskel}.
This platform is parameterized by a \textit{platform depth}
$k = \pdep$, which is the number of microarchitectural buffers.
The non-interference property specifies $\abuf_0$ as the secret input,
and $\abuf_{\pdep-1}$ as public output.
We run our pattern generation algorithm on this platform with varying
values of $\pdep$ and corresponding grammar depths 
$\gdep \in \{\pdep, \pdep+1, \pdep+2\}$.

Fig. \ref{fig:synthetic-generation}(d) presents the results.
We observe that for the case where $\gdep = \pdep$, 
we generate only one pattern with skeleton 
$\anop_1 \cdot \anop_2 \cdots \anop_{\pdep}$
and constraint $\land_i \datadep(\anop_i, \anop_{i+1})$
signifying data propagation through these operations.
Higher grammar depths, however, require multiple counterfactuals
(Ex. \ref{ex:platsynthmulticounterfactual}), and produce more patterns.
In the $\gdep = \pdep+1$ case, we require 2 counterfactuals
at a time, while in the $\gdep = \pdep+2$ case we require
3 counterfactuals at-a-time.
While this leads to a blowup in generation time as evidenced in the figure,
the times are still reasonable ($\sim$10m even in the $\pdep+2$ case).
Moreover we observe that in the realistic platforms considered in previous sections, 
smaller depths were still able to provide us with useful patterns.

\begin{figure}
    \centering
    \includegraphics[scale=0.11]{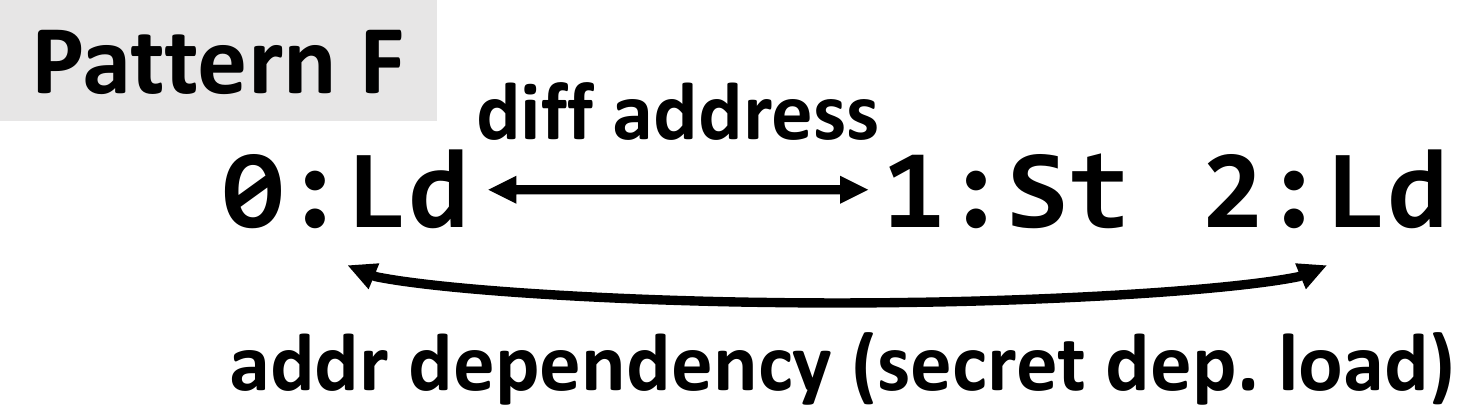}
    \hspace{0.4cm}
    \includegraphics[scale=0.11]{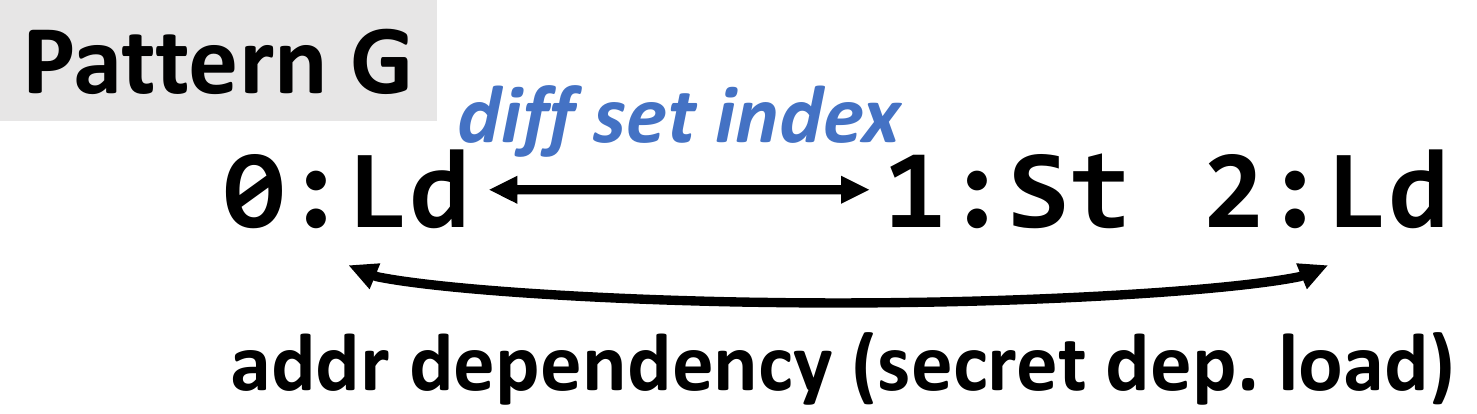}
    \caption{(F) Pattern using default grammar (Tab. \ref{tab:facetgrammar}).
        (G) More precise pattern after adding the $\mathsf{diffindex}$ predicate.}
    \label{fig:coarse-fine-patterns}
\end{figure}

\begin{figure}
    \centering
    \includegraphics[scale=0.27]{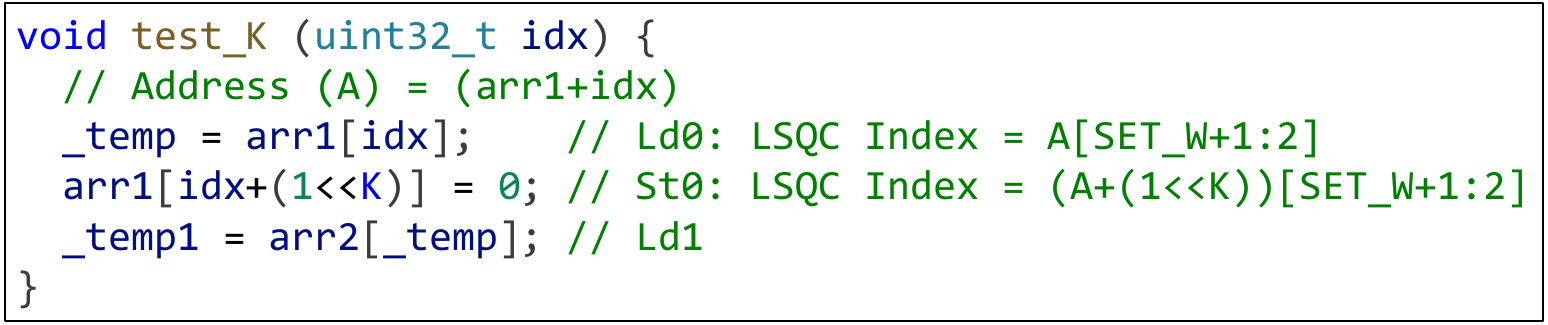}
    \caption{Example illustrating false positives with patterns.}
    \label{fig:fpr-example}
\end{figure}

\subsection{RQ4: How does the choice of grammar affect false positives?}
\label{subsec:rq4}

To explore this, we consider a variant of the $\platss$ platform.
As in the original model (Fig. \ref{fig:platsslisting}), 
a load is sourced from the LSQ cache if there exists
a valid entry with the same address. 
Otherwise, it is sourced from memory and updates the LSQ cache. 
However, in this variant, a store invalidates all 
LSQ entries with the same set index (to avoid loading stale values).
We perform pattern generation with the memory access count 
($\pseudocb{lscount}$) as public output, with grammar depth 3.
We get Pattern F (Fig. \ref{fig:coarse-fine-patterns})
with an address dependency between (\texttt{0:Ld}) and (\texttt{2:Ld}),
(as in SpectreV1), and where the intervening store (\texttt{1:St}) \textit{does not} 
invalidate the LSQ entry (captured by $\diffaddr$).
If \texttt{1:St} invalidated the LSQ entry, \texttt{2:Ld} would access
the memory in all executions, leading to no violating behaviours.

This pattern flags the code in Fig. \ref{fig:fpr-example} due to a 
match on instructions \texttt{Ld0}, \texttt{St0}, \texttt{Ld1}.
However, we note that while the address of the store in Fig. \ref{fig:fpr-example}
is different from the first load, the set index can be the same.
This is because the set index is a slice (sub-word) of the address.
Thus, under the condition that $\texttt{K} > \texttt{SET\_W} + 2$,
the store will be in the same set as the first load, and the 
LSQ entry will be invalidated.
Consequently, under this condition, the program is actually safe,
and Pattern F produces a false positive.

To address this, we can add a new predicate $\mathsf{diffindex}$ to the grammar
that checks if the set indices (defined as an address slice) of two addresses  
are different.
Performing generation with this augmented grammar, we obtain Pattern G 
(Fig. \ref{fig:coarse-fine-patterns}).
This pattern only flags $\texttt{test\_K}$ if $\texttt{K} \leq \texttt{SET\_W} + 2$,
and is hence more precise.
We summarize these observations in Tab. \ref{tab:fpr-example}.
Thus, grammars over (high-level) architectural state 
result in patterns with
more false positives. 
Through the grammar, \textit{we expose a precision-complexity tradeoff -
more precise patterns can be generated at the cost of a carefully 
tailored, microarchitecture-specific grammar}.

\begin{table}
    \centering
    \resizebox*{0.45\textwidth}{!}{
    \begin{tabular}{c|c|c}
        \hline
        \multirow{2}{*}{Check} & \multicolumn{2}{c}{Result with \texttt{test\_K} (Fig. \ref{fig:fpr-example}) and \texttt{SET\_W} set index}  \\ \cline{2-3}
        & $\texttt{K} > \texttt{SET\_W} + 2$ & $\texttt{K} \leq \texttt{SET\_W} + 2$ \\ \hline
        Hyperproperty & SAFE & UNSAFE \\ \hline
        Pat. F (Fig. \ref{fig:coarse-fine-patterns}) & \cellcolor{orange!20} UNSAFE & UNSAFE \\ \hline
        Pat. G (Fig. \ref{fig:coarse-fine-patterns}) & SAFE & UNSAFE \\ 
    \end{tabular}
    }
    \caption{Precision of patterns on \texttt{test\_K} (Fig. \ref{fig:fpr-example}).}
    \label{tab:fpr-example}
    \vspace*{-0.5cm}
\end{table}

\section{Discussion and Limitations}
\label{sec:discussion}

\subsubsection*{Applicability and Scope}

We have applied pattern generation to speculative (e.g., \S\ref{subsec:rq1}, Fig. \ref{fig:stlpattern}) and non-speculative (e.g., \S\ref{subsec:rq3}, Fig. \ref{fig:coarse-fine-patterns}) 
execution, and multiple side channels 
(\S\ref{subsec:platco}, \S\ref{subsec:platss}).
Crucially, this is possible due to the power of 
non-interference-based specifications: 
they can express varied attacker scenarios 
(e.g., secure-programming, constant-time 
\cite{Guarnieri2021HardwareSoftwareCF, Cauligi2019ConstanttimeFF}) and
platform semantics.
While our approach is general, it requires formulating non-interference specifications that accurately captures the threat model.
This requires care and expertise.
The effort in writing these specifications can be amortized over many different patterns generated and corresponding verification and exploit-finding tasks.

\subsubsection*{Platform complexity}
In \S\ref{sec:evaluation} we considered abstract platform models
(introduced in \S\ref{sec:methodology}) that are similar to those adopted in 
$\semhyper$ approaches (e.g., \cite{Cheang2019AFA,
Guarnieri2020SpectectorPD,Zeng2022AutomaticGO}).
These models, while exposing security relevant microarchitectural 
detail, are simpler than full processor RTL.
While the latter being more complex will likely result in larger generation 
times, this is not a fundamental limitation of our approach.
Our main requirement is that the platform model be amenable to
symbolic analysis. This can be performed on RTL 
using off-the-shelf model checkers (e.g., JasperGold, SymbiYosys 
\cite{symbiyosys}).
Additionally, there is work that extracts abstract 
models from RTL (\cite{Zeng2022AutomaticGO,Godbole24Lifting}).

\subsubsection*{Pattern Grammar and False Positives}
While patterns using architecture level predicates
are more robust/microarchitecture-independent,
they lead to more false positives.
As demonstrated in \S\ref{subsec:rq4}, using specialized grammars 
that expose microarchitectural execution details
(e.g., cache-indexing, replacement/prediction policies)
helps reduce false positives.
However, this is at the cost of (a) less performant analysis
(it is abstraction that makes pattern-based approaches scale)
and (b) more microarchitecture-dependent patterns.
Opinion in literature is divided on whether or not to expose 
microarchitectural detail to software analyses 
(e.g., \cite{Cauligi2021SoKPF}).
We view the grammar as a tradeoff: 
a generic grammar leads to abstract and thus efficient analysis
at the cost of more false positives, while a specialized grammar
is more precise at the cost of being more microarchitecture-specific.
While our current approach requires a user-specified grammar 
(Fig. \ref{fig:sempatpipeline}),
future work could automate this, e.g., 
by using counter-example guided predicate discovery \cite{Das2002CounterExampleBP, Ball2001AutomaticPA}.

\section{Related Work}
\label{sec:related}

\subsubsection*{Microarchitectural optimizations and vulnerabilities}

Our work targets hardware execution vulnerabilities, from
Spectre and Meltdown and their variants 
\cite{Kocher2019SpectreAE,Lipp2018MeltdownRK,Schwarz2019NetSpectreRA,Koruyeh2018SpectreRS}
to more recent attacks
(e.g., targeting store/line-fill buffers 
\cite{Canella2019FalloutLD,Schaik2019RIDLRI,Schwarz2019ZombieLoadCD,Maisuradze2018ret2specSE}).
Good overviews of these vulnerabilities may be found in~\cite{Hu2021SoKHD,Ge2018ASO,Canella2018ASE}.
These vulnerabilities are possible because of microarchitectural optimizations such as
instruction reuse \cite{Sodani1997DynamicIR} and silent stores \cite{Lepak2001SilentSA,Lepak2000SilentSF}
which we have used in our experiments.
Attacks exploiting these mechanisms were hypothesized in \cite{Vicarte2021OpeningPB} amongst which
some have recently been demonstrated on actual hardware (e.g., \cite{moghimi2023downfall}).

\subsubsection*{Software Analysis} 
Several approaches perform software analysis,
adopting different specifications, models, and techniques.

\textit{Semantic hyperproperty-based verification.}
Semantic hyperproperty-based approaches
(e.g., \cite{Cheang2019AFA,Guarnieri2020SpectectorPD,Balliu2020InSpectreBA,
Fabian2022AutomaticDO}) formulate security as a hyperproperty
\cite{Terauchi2005SecureIF,Clarkson2008Hyperproperties,Rushby1992NoninterferenceTT,Kozyri2022ExpressingIF}
over executions of the program on an abstract platform model.
These serve as inputs to our approach.

\textit{Symbolic software analysis.}
While the earlier works develop high-level hyperproperties that
capture varied microarchitectural mechanisms, other approaches
adopt more specific security models (e.g., constant-time execution).
This allows them to develop specialized and more efficient analysis techniques
(see \cite{Geimer2023SystematicEO} for a systematic evaluation and comparison),
based on symbolic execution (SE) \cite{King1976SymbolicEA, Wang2019KLEESpectre,Guo2019SPECUSYMSS,David2016BINSECSEAD} or 
\textit{relational} symbolic execution (RelSE) \cite{Farina2017RelationalSE, Daniel2019BinsecRelER,Daniel2021HuntingTH}.
Our work is orthogonal to these: (a) our main goal is pattern generation, not binary analysis and 
(b) we are not limited to a specific (e.g., constant-time) leakage model.

\textit{Pattern-based detection.}
Previous work performing pattern-based detection 
(e.g., \cite{Mosier2022AxiomaticHC, Len2021CatsVS, Trippel2019SecurityVV})
manually defines patterns.
We develop an approach to systematically generate patterns
that are complete (up to a skeleton size),
thus complementing these approaches.

\subsubsection*{Microarchitectural verification/abstraction}

The patterns that we generate can be thought of as software-side 
abstractions extracted based on our semantic platform analysis.

\textit{Contract-based abstractions.}
Approaches (e.g., \cite{Guarnieri2021HardwareSoftwareCF, wang2023specification}) 
develop security contracts at the hardware/software interface.
By proving that (a) the hardware refines the contract and (b) software satisfies the contract,
these approaches guarantee software security.
While the contract can also be viewed as a hardware/software abstraction, 
it is of a different nature than the patterns we generate.

\textit{Security verification/side-channel analysis.}
Several approaches directly verify hardware RTL w.r.t. security properties either 
formally (using symbolic techniques) 
\cite{Fadiheh2018ProcessorHS,Gleissenthall2019IODINEVC,Gleissenthall2021SolverAidedCH}
or using faster but incomplete techniques (e.g., fuzzing) 
\cite{Nilizadeh2018DifFuzzDF, Noller2021QFuzzQF, Oleksenko2021RevizorTB, Tyagi2022TheHuzzIF}.
Finally, there are approaches that perform automated extraction/validation 
of side-channels with white/black box designs \cite{oleksenko2023hide, Nemati2020ValidationOA, Gras2020ABSyntheAB, Weber2021OsirisAD}.
While this too can be seen as extracting security relevant hardware abstractions, 
it needs to be paired with software analysis techniques to be useful for software security verification.
\section{Conclusion}
\label{sec:conclusion}

In this work, we presented an approach to convert a given platform model and 
non-interference-based security hyperproperty into a set of attack patterns.
Our automatic generation approach improves on manual pattern creation 
(which can result in missed patterns) by guaranteeing that 
the generated patterns capture all non-interference violations up to a certain size.
We implemented our approach in a prototype tool.
Our evaluation resulted in the identification of, to our knowledge,
previously unknown patterns for Spectre BCB and Spectre STL, and 
other vulnerability variants.
We also demonstrated 
improved verification performance using generated patterns as compared to the original hyperproperty.
By providing a systematic way to generate patterns, our work
combines the best of both worlds: formal guarantees of hyperproperty-based
specification and scalability of pattern-based verification.

\bibliographystyle{ACM-Reference-Format}
\bibliography{refs}


\begin{thebibliography}{76}


\ifx \showCODEN    \undefined \def \showCODEN     #1{\unskip}     \fi
\ifx \showDOI      \undefined \def \showDOI       #1{#1}\fi
\ifx \showISBNx    \undefined \def \showISBNx     #1{\unskip}     \fi
\ifx \showISBNxiii \undefined \def \showISBNxiii  #1{\unskip}     \fi
\ifx \showISSN     \undefined \def \showISSN      #1{\unskip}     \fi
\ifx \showLCCN     \undefined \def \showLCCN      #1{\unskip}     \fi
\ifx \shownote     \undefined \def \shownote      #1{#1}          \fi
\ifx \showarticletitle \undefined \def \showarticletitle #1{#1}   \fi
\ifx \showURL      \undefined \def \showURL       {\relax}        \fi
\providecommand\bibfield[2]{#2}
\providecommand\bibinfo[2]{#2}
\providecommand\natexlab[1]{#1}
\providecommand\showeprint[2][]{arXiv:#2}

\bibitem[Alglave et~al\mbox{.}(2014)]%
        {Alglave2014HerdingCM}
\bibfield{author}{\bibinfo{person}{Jade Alglave}, \bibinfo{person}{Luc
  Maranget}, {and} \bibinfo{person}{Michael Tautschnig}.}
  \bibinfo{year}{2014}\natexlab{}.
\newblock \showarticletitle{Herding cats: modelling, simulation, testing, and
  data-mining for weak memory}.
\newblock \bibinfo{journal}{\emph{Proceedings of the 35th ACM SIGPLAN
  Conference on Programming Language Design and Implementation}}
  (\bibinfo{year}{2014}).
\newblock


\bibitem[Alur et~al\mbox{.}(2013)]%
        {Alur2013SyntaxguidedS}
\bibfield{author}{\bibinfo{person}{Rajeev Alur}, \bibinfo{person}{Rastislav
  Bodik}, \bibinfo{person}{Garvit Juniwal}, \bibinfo{person}{Milo M.~K.
  Martin}, \bibinfo{person}{Mukund Raghothaman}, \bibinfo{person}{Sanjit~A.
  Seshia}, \bibinfo{person}{Rishabh Singh}, \bibinfo{person}{Armando
  Solar-Lezama}, \bibinfo{person}{Emina Torlak}, {and}
  \bibinfo{person}{Abhishek Udupa}.} \bibinfo{year}{2013}\natexlab{}.
\newblock \showarticletitle{Syntax-guided synthesis}. In
  \bibinfo{booktitle}{\emph{2013 Formal Methods in Computer-Aided Design}}.
  \bibinfo{publisher}{IEEE}, \bibinfo{pages}{1--8}.
\newblock
\urldef\tempurl%
\url{https://doi.org/10.1109/fmcad.2013.6679385}
\showDOI{\tempurl}


\bibitem[Ball et~al\mbox{.}(2001)]%
        {Ball2001AutomaticPA}
\bibfield{author}{\bibinfo{person}{Thomas Ball}, \bibinfo{person}{Rupak
  Majumdar}, \bibinfo{person}{Todd Millstein}, {and} \bibinfo{person}{Sriram~K.
  Rajamani}.} \bibinfo{year}{2001}\natexlab{}.
\newblock \showarticletitle{Automatic predicate abstraction of C programs}. In
  \bibinfo{booktitle}{\emph{Proceedings of the ACM SIGPLAN 2001 Conference on
  Programming Language Design and Implementation}} (Snowbird, Utah, USA)
  \emph{(\bibinfo{series}{PLDI '01})}. \bibinfo{publisher}{Association for
  Computing Machinery}, \bibinfo{address}{New York, NY, USA},
  \bibinfo{pages}{203--213}.
\newblock
\showISBNx{1581134142}
\urldef\tempurl%
\url{https://doi.org/10.1145/378795.378846}
\showDOI{\tempurl}


\bibitem[Barbosa et~al\mbox{.}(2022)]%
        {Barbosa2022cvc5AV}
\bibfield{author}{\bibinfo{person}{Haniel Barbosa}, \bibinfo{person}{Clark~W.
  Barrett}, \bibinfo{person}{Martin Brain}, \bibinfo{person}{Gereon Kremer},
  \bibinfo{person}{Hanna Lachnitt}, \bibinfo{person}{Makai Mann},
  \bibinfo{person}{Abdalrhman Mohamed}, \bibinfo{person}{Mudathir Mohamed},
  \bibinfo{person}{Aina Niemetz}, \bibinfo{person}{Andres N{\"o}tzli},
  \bibinfo{person}{Alex Ozdemir}, \bibinfo{person}{Mathias Preiner},
  \bibinfo{person}{Andrew Reynolds}, \bibinfo{person}{Ying Sheng},
  \bibinfo{person}{Cesare Tinelli}, {and} \bibinfo{person}{Yoni Zohar}.}
  \bibinfo{year}{2022}\natexlab{}.
\newblock \showarticletitle{cvc5: A Versatile and Industrial-Strength SMT
  Solver}. In \bibinfo{booktitle}{\emph{TACAS}}.
\newblock


\bibitem[Barrett et~al\mbox{.}(2016)]%
        {BarFT-SMTLIB}
\bibfield{author}{\bibinfo{person}{Clark Barrett}, \bibinfo{person}{Pascal
  Fontaine}, {and} \bibinfo{person}{Cesare Tinelli}.}
  \bibinfo{year}{2016}\natexlab{}.
\newblock \bibinfo{title}{{The Satisfiability Modulo Theories Library
  (SMT-LIB)}}.
\newblock \bibinfo{howpublished}{{\tt www.SMT-LIB.org}}.
\newblock


\bibitem[Barrett et~al\mbox{.}(2009)]%
        {Barrett2009SatisfiabilityMT}
\bibfield{author}{\bibinfo{person}{Clark~W. Barrett}, \bibinfo{person}{Roberto
  Sebastiani}, \bibinfo{person}{Sanjit~A. Seshia}, {and}
  \bibinfo{person}{Cesare Tinelli}.} \bibinfo{year}{2009}\natexlab{}.
\newblock \showarticletitle{Satisfiability Modulo Theories}. In
  \bibinfo{booktitle}{\emph{Handbook of Satisfiability}}.
\newblock


\bibitem[Canella et~al\mbox{.}(2019a)]%
        {Canella2018ASE}
\bibfield{author}{\bibinfo{person}{Claudio Canella}, \bibinfo{person}{Jo~Van
  Bulck}, \bibinfo{person}{Michael Schwarz}, \bibinfo{person}{Moritz Lipp},
  \bibinfo{person}{Benjamin von Berg}, \bibinfo{person}{Philipp Ortner},
  \bibinfo{person}{Frank Piessens}, \bibinfo{person}{Dmitry Evtyushkin}, {and}
  \bibinfo{person}{Daniel Gruss}.} \bibinfo{year}{2019}\natexlab{a}.
\newblock \showarticletitle{A Systematic Evaluation of Transient Execution
  Attacks and Defenses}.
\newblock  (\bibinfo{date}{Aug.} \bibinfo{year}{2019}),
  \bibinfo{pages}{249--266}.
\newblock
\showISBNx{978-1-939133-06-9}
\urldef\tempurl%
\url{https://www.usenix.org/conference/usenixsecurity19/presentation/canella}
\showURL{%
\tempurl}


\bibitem[Canella et~al\mbox{.}(2019b)]%
        {Canella2019FalloutLD}
\bibfield{author}{\bibinfo{person}{Claudio Canella}, \bibinfo{person}{Daniel
  Genkin}, \bibinfo{person}{Lukas Giner}, \bibinfo{person}{Daniel Gruss},
  \bibinfo{person}{Moritz Lipp}, \bibinfo{person}{Marina Minkin},
  \bibinfo{person}{Daniel Moghimi}, \bibinfo{person}{Frank Piessens},
  \bibinfo{person}{Michael Schwarz}, \bibinfo{person}{Berk Sunar},
  \bibinfo{person}{Jo~Van Bulck}, {and} \bibinfo{person}{Yuval Yarom}.}
  \bibinfo{year}{2019}\natexlab{b}.
\newblock \showarticletitle{Fallout: Leaking Data on Meltdown-resistant CPUs}.
\newblock \bibinfo{journal}{\emph{Proceedings of the 2019 ACM SIGSAC Conference
  on Computer and Communications Security}} (\bibinfo{year}{2019}).
\newblock


\bibitem[Cauligi et~al\mbox{.}(2020)]%
        {Cauligi2019ConstanttimeFF}
\bibfield{author}{\bibinfo{person}{Sunjay Cauligi}, \bibinfo{person}{Craig
  Disselkoen}, \bibinfo{person}{Klaus~v. Gleissenthall}, \bibinfo{person}{Dean
  Tullsen}, \bibinfo{person}{Deian Stefan}, \bibinfo{person}{Tamara Rezk},
  {and} \bibinfo{person}{Gilles Barthe}.} \bibinfo{year}{2020}\natexlab{}.
\newblock \showarticletitle{Constant-time foundations for the new spectre era}.
  In \bibinfo{booktitle}{\emph{Proceedings of the 41st ACM SIGPLAN Conference
  on Programming Language Design and Implementation}}
  \emph{(\bibinfo{series}{PLDI '20})}. \bibinfo{publisher}{ACM}.
\newblock
\urldef\tempurl%
\url{https://doi.org/10.1145/3385412.3385970}
\showDOI{\tempurl}


\bibitem[Cauligi et~al\mbox{.}(2022)]%
        {Cauligi2021SoKPF}
\bibfield{author}{\bibinfo{person}{Sunjay Cauligi}, \bibinfo{person}{Craig
  Disselkoen}, \bibinfo{person}{Daniel Moghimi}, \bibinfo{person}{Gilles
  Barthe}, {and} \bibinfo{person}{Deian Stefan}.}
  \bibinfo{year}{2022}\natexlab{}.
\newblock \showarticletitle{SoK: Practical Foundations for Software Spectre
  Defenses}. In \bibinfo{booktitle}{\emph{2022 IEEE Symposium on Security and
  Privacy (SP)}}. \bibinfo{publisher}{IEEE}, \bibinfo{pages}{666--680}.
\newblock
\urldef\tempurl%
\url{https://doi.org/10.1109/sp46214.2022.9833707}
\showDOI{\tempurl}


\bibitem[Charles et~al\mbox{.}(2002)]%
        {Charles2002ImprovingIW}
\bibfield{author}{\bibinfo{person}{Desiree Charles}, \bibinfo{person}{Ali~R.
  Hurson}, {and} \bibinfo{person}{Narayanan Vijaykrishnan}.}
  \bibinfo{year}{2002}\natexlab{}.
\newblock \showarticletitle{Improving ILP with instruction-reuse cache
  hierarchy}.
\newblock \bibinfo{journal}{\emph{Fifth International Conference on Algorithms
  and Architectures for Parallel Processing, 2002. Proceedings.}}
  (\bibinfo{year}{2002}), \bibinfo{pages}{206--213}.
\newblock


\bibitem[Cheang et~al\mbox{.}(2019)]%
        {Cheang2019AFA}
\bibfield{author}{\bibinfo{person}{Kevin Cheang}, \bibinfo{person}{Cameron
  Rasmussen}, \bibinfo{person}{Sanjit Seshia}, {and} \bibinfo{person}{Pramod
  Subramanyan}.} \bibinfo{year}{2019}\natexlab{}.
\newblock \showarticletitle{A Formal Approach to Secure Speculation}. In
  \bibinfo{booktitle}{\emph{2019 IEEE 32nd Computer Security Foundations
  Symposium (CSF)}}. \bibinfo{publisher}{IEEE}, \bibinfo{pages}{288--28815}.
\newblock
\urldef\tempurl%
\url{https://doi.org/10.1109/csf.2019.00027}
\showDOI{\tempurl}


\bibitem[{Claire Wolf, et. al.}(2022)]%
        {symbiyosys}
\bibfield{author}{\bibinfo{person}{{Claire Wolf, et. al.}}}
  \bibinfo{year}{2022}\natexlab{}.
\newblock \bibinfo{title}{SymbiYosys}.
\newblock \bibinfo{howpublished}{\url{https://github.com/YosysHQ/sby}}.
\newblock


\bibitem[Clarke et~al\mbox{.}(1994)]%
        {Clarke1993ModelCA}
\bibfield{author}{\bibinfo{person}{Edmund~M. Clarke}, \bibinfo{person}{Orna
  Grumberg}, {and} \bibinfo{person}{David~E. Long}.}
  \bibinfo{year}{1994}\natexlab{}.
\newblock \showarticletitle{Model checking and abstraction}.
\newblock \bibinfo{journal}{\emph{ACM Trans. Program. Lang. Syst.}}
  \bibinfo{volume}{16}, \bibinfo{number}{5} (\bibinfo{date}{Sep.}
  \bibinfo{year}{1994}), \bibinfo{pages}{1512--1542}.
\newblock
\showISSN{0164-0925}
\urldef\tempurl%
\url{https://doi.org/10.1145/186025.186051}
\showDOI{\tempurl}


\bibitem[Clarkson and Schneider(2008)]%
        {Clarkson2008Hyperproperties}
\bibfield{author}{\bibinfo{person}{Michael~R. Clarkson} {and}
  \bibinfo{person}{Fred~B. Schneider}.} \bibinfo{year}{2008}\natexlab{}.
\newblock \showarticletitle{Hyperproperties}.
\newblock \bibinfo{journal}{\emph{2008 21st IEEE Computer Security Foundations
  Symposium}} (\bibinfo{year}{2008}), \bibinfo{pages}{51--65}.
\newblock


\bibitem[Daniel et~al\mbox{.}(2021)]%
        {Daniel2021HuntingTH}
\bibfield{author}{\bibinfo{person}{Lesly-Ann Daniel},
  \bibinfo{person}{S{\'e}bastien Bardin}, {and} \bibinfo{person}{Tamara Rezk}.}
  \bibinfo{year}{2021}\natexlab{}.
\newblock \showarticletitle{Hunting the Haunter - Efficient Relational Symbolic
  Execution for Spectre with Haunted RelSE}.
\newblock \bibinfo{journal}{\emph{Proceedings 2021 Network and Distributed
  System Security Symposium}} (\bibinfo{year}{2021}).
\newblock
\urldef\tempurl%
\url{https://api.semanticscholar.org/CorpusID:231878700}
\showURL{%
\tempurl}


\bibitem[Daniel et~al\mbox{.}(2023)]%
        {Daniel2019BinsecRelER}
\bibfield{author}{\bibinfo{person}{Lesly-Ann Daniel},
  \bibinfo{person}{S\'{e}bastien Bardin}, {and} \bibinfo{person}{Tamara Rezk}.}
  \bibinfo{year}{2023}\natexlab{}.
\newblock \showarticletitle{Binsec/Rel: Symbolic Binary Analyzer for Security
  with Applications to Constant-Time and Secret-Erasure}.
\newblock \bibinfo{journal}{\emph{ACM Trans. Priv. Secur.}}
  \bibinfo{volume}{26}, \bibinfo{number}{2}, Article \bibinfo{articleno}{11}
  (\bibinfo{date}{Apr} \bibinfo{year}{2023}), \bibinfo{numpages}{42}~pages.
\newblock
\showISSN{2471-2566}
\urldef\tempurl%
\url{https://doi.org/10.1145/3563037}
\showDOI{\tempurl}


\bibitem[Das and Dill(2002)]%
        {Das2002CounterExampleBP}
\bibfield{author}{\bibinfo{person}{Satyaki Das} {and} \bibinfo{person}{David~L.
  Dill}.} \bibinfo{year}{2002}\natexlab{}.
\newblock \showarticletitle{Counter-Example Based Predicate Discovery in
  Predicate Abstraction}. In \bibinfo{booktitle}{\emph{Formal Methods in
  Computer-Aided Design}}, \bibfield{editor}{\bibinfo{person}{Mark~D. Aagaard}
  {and} \bibinfo{person}{John~W. O'Leary}} (Eds.). \bibinfo{publisher}{Springer
  Berlin Heidelberg}, \bibinfo{address}{Berlin, Heidelberg},
  \bibinfo{pages}{19--32}.
\newblock
\showISBNx{978-3-540-36126-8}
\urldef\tempurl%
\url{https://doi.org/10.1145/378795.378846}
\showDOI{\tempurl}


\bibitem[David et~al\mbox{.}(2016)]%
        {David2016BINSECSEAD}
\bibfield{author}{\bibinfo{person}{Robin David}, \bibinfo{person}{S{\'e}bastien
  Bardin}, \bibinfo{person}{Thanh~Dinh Ta}, \bibinfo{person}{Laurent Mounier},
  \bibinfo{person}{Josselin Feist}, \bibinfo{person}{Marie-Laure Potet}, {and}
  \bibinfo{person}{Jean-Yves Marion}.} \bibinfo{year}{2016}\natexlab{}.
\newblock \showarticletitle{BINSEC/SE: A Dynamic Symbolic Execution Toolkit for
  Binary-Level Analysis}.
\newblock \bibinfo{journal}{\emph{2016 IEEE 23rd International Conference on
  Software Analysis, Evolution, and Reengineering (SANER)}}
  \bibinfo{volume}{1} (\bibinfo{year}{2016}), \bibinfo{pages}{653--656}.
\newblock
\urldef\tempurl%
\url{https://api.semanticscholar.org/CorpusID:7488274}
\showURL{%
\tempurl}


\bibitem[de~Moura and Bj{\o}rner(2008)]%
        {Moura2008Z3AE}
\bibfield{author}{\bibinfo{person}{Leonardo~Mendonça de Moura} {and}
  \bibinfo{person}{Nikolaj~S. Bj{\o}rner}.} \bibinfo{year}{2008}\natexlab{}.
\newblock \showarticletitle{Z3: An Efficient SMT Solver}. In
  \bibinfo{booktitle}{\emph{International Conference on Tools and Algorithms
  for Construction and Analysis of Systems}}.
\newblock


\bibitem[Fabian et~al\mbox{.}(2022)]%
        {Fabian2022AutomaticDO}
\bibfield{author}{\bibinfo{person}{Xaver Fabian}, \bibinfo{person}{Marco
  Guarnieri}, {and} \bibinfo{person}{Marco Patrignani}.}
  \bibinfo{year}{2022}\natexlab{}.
\newblock \showarticletitle{Automatic Detection of Speculative Execution
  Combinations}.
\newblock \bibinfo{journal}{\emph{Proceedings of the 2022 ACM SIGSAC Conference
  on Computer and Communications Security}} (\bibinfo{year}{2022}).
\newblock


\bibitem[Fadiheh et~al\mbox{.}(2018)]%
        {Fadiheh2018ProcessorHS}
\bibfield{author}{\bibinfo{person}{Mohammad~Rahmani Fadiheh},
  \bibinfo{person}{Dominik Stoffel}, \bibinfo{person}{Clark~W. Barrett},
  \bibinfo{person}{Subhasish Mitra}, {and} \bibinfo{person}{Wolfgang Kunz}.}
  \bibinfo{year}{2018}\natexlab{}.
\newblock \showarticletitle{Processor Hardware Security Vulnerabilities and
  their Detection by Unique Program Execution Checking}.
\newblock \bibinfo{journal}{\emph{2019 Design, Automation \& Test in Europe
  Conference \& Exhibition (DATE)}} (\bibinfo{year}{2018}),
  \bibinfo{pages}{994--999}.
\newblock


\bibitem[Farina et~al\mbox{.}(2019)]%
        {Farina2017RelationalSE}
\bibfield{author}{\bibinfo{person}{Gian~Pietro Farina},
  \bibinfo{person}{Stephen Chong}, {and} \bibinfo{person}{Marco Gaboardi}.}
  \bibinfo{year}{2019}\natexlab{}.
\newblock \showarticletitle{Relational Symbolic Execution}. In
  \bibinfo{booktitle}{\emph{Proceedings of the 21st International Symposium on
  Principles and Practice of Declarative Programming}}
  \emph{(\bibinfo{series}{PPDP '19})}. \bibinfo{publisher}{ACM}.
\newblock
\urldef\tempurl%
\url{https://doi.org/10.1145/3354166.3354175}
\showDOI{\tempurl}


\bibitem[Ge et~al\mbox{.}(2018)]%
        {Ge2018ASO}
\bibfield{author}{\bibinfo{person}{Qian Ge}, \bibinfo{person}{Yuval Yarom},
  \bibinfo{person}{David~A. Cock}, {and} \bibinfo{person}{Gernot Heiser}.}
  \bibinfo{year}{2018}\natexlab{}.
\newblock \showarticletitle{A survey of microarchitectural timing attacks and
  countermeasures on contemporary hardware}.
\newblock \bibinfo{journal}{\emph{J. Cryptogr. Eng.}} \bibinfo{volume}{8},
  \bibinfo{number}{1} (\bibinfo{year}{2018}), \bibinfo{pages}{1--27}.
\newblock
\urldef\tempurl%
\url{https://doi.org/10.1007/S13389-016-0141-6}
\showDOI{\tempurl}


\bibitem[Geimer et~al\mbox{.}(2023)]%
        {Geimer2023SystematicEO}
\bibfield{author}{\bibinfo{person}{Antoine Geimer}, \bibinfo{person}{Math\'{e}o
  Vergnolle}, \bibinfo{person}{Fr\'{e}d\'{e}ric Recoules},
  \bibinfo{person}{Lesly-Ann Daniel}, \bibinfo{person}{S\'{e}bastien Bardin},
  {and} \bibinfo{person}{Cl\'{e}mentine Maurice}.}
  \bibinfo{year}{2023}\natexlab{}.
\newblock \showarticletitle{A Systematic Evaluation of Automated Tools for
  Side-Channel Vulnerabilities Detection in Cryptographic Libraries}. In
  \bibinfo{booktitle}{\emph{Proceedings of the 2023 ACM SIGSAC Conference on
  Computer and Communications Security}} (Copenhagen, Denmark)
  \emph{(\bibinfo{series}{CCS '23})}. \bibinfo{publisher}{Association for
  Computing Machinery}, \bibinfo{address}{New York, NY, USA},
  \bibinfo{pages}{1690--1704}.
\newblock
\showISBNx{9798400700507}
\urldef\tempurl%
\url{https://doi.org/10.1145/3576915.3623112}
\showDOI{\tempurl}


\bibitem[Godbole et~al\mbox{.}(2024)]%
        {Godbole24Lifting}
\bibfield{author}{\bibinfo{person}{Adwait Godbole}, \bibinfo{person}{Kevin
  Cheang}, \bibinfo{person}{Yatin~A. Manerkar}, {and}
  \bibinfo{person}{Sanjit~A. Seshia}.} \bibinfo{year}{2024}\natexlab{}.
\newblock \showarticletitle{Lifting Micro-Update Models from RTL for Formal
  Security Analysis}. In \bibinfo{booktitle}{\emph{Proceedings of the 29th ACM
  International Conference on Architectural Support for Programming Languages
  and Operating Systems, Volume 2}} (La Jolla, CA, USA)
  \emph{(\bibinfo{series}{ASPLOS '24})}. \bibinfo{publisher}{Association for
  Computing Machinery}, \bibinfo{address}{New York, NY, USA},
  \bibinfo{pages}{631--648}.
\newblock
\showISBNx{9798400703850}
\urldef\tempurl%
\url{https://doi.org/10.1145/3620665.3640418}
\showDOI{\tempurl}


\bibitem[Goguen and Meseguer(1984)]%
        {Goguen1984UnwindingAI}
\bibfield{author}{\bibinfo{person}{Joseph~A. Goguen} {and}
  \bibinfo{person}{Jos{\'e} Meseguer}.} \bibinfo{year}{1984}\natexlab{}.
\newblock \showarticletitle{Unwinding and Inference Control}.
\newblock \bibinfo{journal}{\emph{1984 IEEE Symposium on Security and Privacy}}
  (\bibinfo{year}{1984}), \bibinfo{pages}{75--75}.
\newblock


\bibitem[Gras et~al\mbox{.}(2020)]%
        {Gras2020ABSyntheAB}
\bibfield{author}{\bibinfo{person}{Ben Gras}, \bibinfo{person}{Cristiano
  Giuffrida}, \bibinfo{person}{Michael Kurth}, \bibinfo{person}{Herbert Bos},
  {and} \bibinfo{person}{Kaveh Razavi}.} \bibinfo{year}{2020}\natexlab{}.
\newblock \showarticletitle{{ABSynthe: Automatic Blackbox Side-channel
  Synthesis on Commodity Microarchitectures}}. In
  \bibinfo{booktitle}{\emph{27th Annual Network and Distributed System Security
  Symposium, {NDSS} 2020, San Diego, California, USA, February 23-26, 2020}}.
  \bibinfo{publisher}{The Internet Society}.
\newblock
\urldef\tempurl%
\url{https://www.ndss-symposium.org/ndss-paper/absynthe-automatic-blackbox-side-channel-synthesis-on-commodity-microarchitectures/}
\showURL{%
\tempurl}


\bibitem[Guanciale et~al\mbox{.}(2020)]%
        {Balliu2020InSpectreBA}
\bibfield{author}{\bibinfo{person}{Roberto Guanciale}, \bibinfo{person}{Musard
  Balliu}, {and} \bibinfo{person}{Mads Dam}.} \bibinfo{year}{2020}\natexlab{}.
\newblock \showarticletitle{InSpectre: Breaking and Fixing Microarchitectural
  Vulnerabilities by Formal Analysis}. In \bibinfo{booktitle}{\emph{Proceedings
  of the 2020 ACM SIGSAC Conference on Computer and Communications Security}}
  \emph{(\bibinfo{series}{CCS '20})}. \bibinfo{publisher}{ACM}.
\newblock
\urldef\tempurl%
\url{https://doi.org/10.1145/3372297.3417246}
\showDOI{\tempurl}


\bibitem[Guarnieri et~al\mbox{.}(2020)]%
        {Guarnieri2020SpectectorPD}
\bibfield{author}{\bibinfo{person}{Marco Guarnieri}, \bibinfo{person}{Boris
  K{\"o}pf}, \bibinfo{person}{Jos{\'e}~Francisco Morales}, \bibinfo{person}{Jan
  Reineke}, {and} \bibinfo{person}{Andr{\'e}s S{\'a}nchez}.}
  \bibinfo{year}{2020}\natexlab{}.
\newblock \showarticletitle{Spectector: Principled Detection of Speculative
  Information Flows}. In \bibinfo{booktitle}{\emph{2020 IEEE Symposium on
  Security and Privacy (SP)}}. \bibinfo{publisher}{IEEE},
  \bibinfo{pages}{1--19}.
\newblock
\urldef\tempurl%
\url{https://doi.org/10.1109/sp40000.2020.00011}
\showDOI{\tempurl}


\bibitem[Guarnieri et~al\mbox{.}(2021)]%
        {Guarnieri2021HardwareSoftwareCF}
\bibfield{author}{\bibinfo{person}{Marco Guarnieri}, \bibinfo{person}{Boris
  K{\"o}pf}, \bibinfo{person}{Jan Reineke}, {and} \bibinfo{person}{Pepe Vila}.}
  \bibinfo{year}{2021}\natexlab{}.
\newblock \showarticletitle{Hardware-Software Contracts for Secure
  Speculation}.
\newblock \bibinfo{journal}{\emph{2021 IEEE Symposium on Security and Privacy
  (SP)}} (\bibinfo{year}{2021}), \bibinfo{pages}{1868--1883}.
\newblock


\bibitem[Guo et~al\mbox{.}(2019)]%
        {Guo2019SPECUSYMSS}
\bibfield{author}{\bibinfo{person}{Shengjian Guo}, \bibinfo{person}{Yueqi
  Chen}, \bibinfo{person}{Peng Li}, \bibinfo{person}{Yueqiang Cheng},
  \bibinfo{person}{Huibo Wang}, \bibinfo{person}{Meng Wu}, {and}
  \bibinfo{person}{Zhiqiang Zuo}.} \bibinfo{year}{2019}\natexlab{}.
\newblock \showarticletitle{SPECUSYM: Speculative Symbolic Execution for Cache
  Timing Leak Detection}.
\newblock \bibinfo{journal}{\emph{2020 IEEE/ACM 42nd International Conference
  on Software Engineering (ICSE)}} (\bibinfo{year}{2019}),
  \bibinfo{pages}{1235--1247}.
\newblock
\urldef\tempurl%
\url{https://api.semanticscholar.org/CorpusID:207869919}
\showURL{%
\tempurl}


\bibitem[Hu et~al\mbox{.}(2021)]%
        {Hu2021SoKHD}
\bibfield{author}{\bibinfo{person}{Guangyuan Hu}, \bibinfo{person}{Zecheng He},
  {and} \bibinfo{person}{Ruby~B. Lee}.} \bibinfo{year}{2021}\natexlab{}.
\newblock \showarticletitle{SoK: Hardware Defenses Against Speculative
  Execution Attacks}. In \bibinfo{booktitle}{\emph{2021 International Symposium
  on Secure and Private Execution Environment Design (SEED)}}.
  \bibinfo{publisher}{IEEE}, \bibinfo{pages}{108--120}.
\newblock
\urldef\tempurl%
\url{https://doi.org/10.1109/seed51797.2021.00023}
\showDOI{\tempurl}


\bibitem[Intel(2023)]%
        {IntelDOIT}
\bibfield{author}{\bibinfo{person}{Intel}.} \bibinfo{year}{2023}\natexlab{}.
\newblock \bibinfo{title}{{Data Operand Independent Timing Instruction Set
  Architecture (ISA) Guidance}}.
\newblock
\newblock
\urldef\tempurl%
\url{https://www.intel.com/content/www/us/en/developer/articles/technical/software-security-guidance/best-practices/data-operand-independent-timing-isa-guidance.html}
\showURL{%
\tempurl}
\newblock
\shownote{Accessed: 2023-11-23}.


\bibitem[Jha and Seshia(2015)]%
        {Jha2015ATO}
\bibfield{author}{\bibinfo{person}{Susmit Jha} {and} \bibinfo{person}{Sanjit~A.
  Seshia}.} \bibinfo{year}{2015}\natexlab{}.
\newblock \showarticletitle{A theory of formal synthesis via inductive
  learning}.
\newblock \bibinfo{journal}{\emph{Acta Informatica}}  \bibinfo{volume}{54}
  (\bibinfo{year}{2015}), \bibinfo{pages}{693--726}.
\newblock


\bibitem[King(1976)]%
        {King1976SymbolicEA}
\bibfield{author}{\bibinfo{person}{James~C. King}.}
  \bibinfo{year}{1976}\natexlab{}.
\newblock \showarticletitle{Symbolic execution and program testing}.
\newblock \bibinfo{journal}{\emph{Commun. ACM}} \bibinfo{volume}{19},
  \bibinfo{number}{7} (\bibinfo{date}{Jul} \bibinfo{year}{1976}),
  \bibinfo{pages}{385--394}.
\newblock
\showISSN{1557-7317}
\urldef\tempurl%
\url{https://doi.org/10.1145/360248.360252}
\showDOI{\tempurl}


\bibitem[Kocher(2018)]%
        {kocher}
\bibfield{author}{\bibinfo{person}{Paul Kocher}.}
  \bibinfo{year}{2018}\natexlab{}.
\newblock \bibinfo{booktitle}{\emph{Spectre Mitigations in Microsoft’s C/C++
  Compiler}}.
\newblock
\urldef\tempurl%
\url{https://www.paulkocher.com/doc/MicrosoftCompilerSpectreMitigation.html}
\showURL{%
\tempurl}


\bibitem[Kocher et~al\mbox{.}(2019)]%
        {Kocher2019SpectreAE}
\bibfield{author}{\bibinfo{person}{Paul~C. Kocher}, \bibinfo{person}{Daniel
  Genkin}, \bibinfo{person}{Daniel Gruss}, \bibinfo{person}{Werner Haas},
  \bibinfo{person}{Michael Hamburg}, \bibinfo{person}{Moritz Lipp},
  \bibinfo{person}{Stefan Mangard}, \bibinfo{person}{Thomas Prescher},
  \bibinfo{person}{Michael Schwarz}, {and} \bibinfo{person}{Yuval Yarom}.}
  \bibinfo{year}{2019}\natexlab{}.
\newblock \showarticletitle{Spectre Attacks: Exploiting Speculative Execution}.
\newblock \bibinfo{journal}{\emph{2019 IEEE Symposium on Security and Privacy
  (SP)}} (\bibinfo{year}{2019}), \bibinfo{pages}{1--19}.
\newblock


\bibitem[Koruyeh et~al\mbox{.}(2024)]%
        {Koruyeh2018SpectreRS}
\bibfield{author}{\bibinfo{person}{Esmaeil~Mohammadian Koruyeh},
  \bibinfo{person}{Khaled~N. Khasawneh}, \bibinfo{person}{Chengyu Song}, {and}
  \bibinfo{person}{Nael Abu-Ghazaleh}.} \bibinfo{year}{2024}\natexlab{}.
\newblock \showarticletitle{Spectre Returns! Speculation Attacks Using the
  Return Stack Buffer}.
\newblock \bibinfo{journal}{\emph{IEEE Design \& Test}} \bibinfo{volume}{41},
  \bibinfo{number}{2}, \bibinfo{pages}{47--55}.
\newblock
\showISSN{2168-2364}
\urldef\tempurl%
\url{https://doi.org/10.1109/mdat.2024.3352537}
\showDOI{\tempurl}


\bibitem[Kozyri et~al\mbox{.}(2022)]%
        {Kozyri2022ExpressingIF}
\bibfield{author}{\bibinfo{person}{Elisavet Kozyri}, \bibinfo{person}{Stephen
  Chong}, {and} \bibinfo{person}{Andrew~C. Myers}.}
  \bibinfo{year}{2022}\natexlab{}.
\newblock \showarticletitle{Expressing Information Flow Properties}.
\newblock \bibinfo{journal}{\emph{Found. Trends Priv. Secur.}}
  \bibinfo{volume}{3} (\bibinfo{year}{2022}), \bibinfo{pages}{1--102}.
\newblock


\bibitem[Lepak et~al\mbox{.}(2001)]%
        {Lepak2001SilentSA}
\bibfield{author}{\bibinfo{person}{Kevin~M. Lepak}, \bibinfo{person}{Gordon~B.
  Bell}, {and} \bibinfo{person}{Mikko~H. Lipasti}.}
  \bibinfo{year}{2001}\natexlab{}.
\newblock \showarticletitle{Silent Stores and Store Value Locality}.
\newblock \bibinfo{journal}{\emph{IEEE Trans. Computers}}  \bibinfo{volume}{50}
  (\bibinfo{year}{2001}), \bibinfo{pages}{1174--1190}.
\newblock


\bibitem[Lepak and Lipasti(2000a)]%
        {Lepak2000OnTV}
\bibfield{author}{\bibinfo{person}{Kevin~M. Lepak} {and}
  \bibinfo{person}{Mikko~H. Lipasti}.} \bibinfo{year}{2000}\natexlab{a}.
\newblock \showarticletitle{On the value locality of store instructions}.
\newblock \bibinfo{journal}{\emph{Proceedings of 27th International Symposium
  on Computer Architecture (IEEE Cat. No.RS00201)}} (\bibinfo{year}{2000}),
  \bibinfo{pages}{182--191}.
\newblock


\bibitem[Lepak and Lipasti(2000b)]%
        {Lepak2000SilentSF}
\bibfield{author}{\bibinfo{person}{Kevin~M. Lepak} {and}
  \bibinfo{person}{Mikko~H. Lipasti}.} \bibinfo{year}{2000}\natexlab{b}.
\newblock \showarticletitle{Silent stores for free}.
\newblock \bibinfo{journal}{\emph{Proceedings 33rd Annual IEEE/ACM
  International Symposium on Microarchitecture. MICRO-33 2000}}
  (\bibinfo{year}{2000}), \bibinfo{pages}{22--31}.
\newblock


\bibitem[Lipp et~al\mbox{.}(2018)]%
        {Lipp2018MeltdownRK}
\bibfield{author}{\bibinfo{person}{Moritz Lipp}, \bibinfo{person}{Michael
  Schwarz}, \bibinfo{person}{Daniel Gruss}, \bibinfo{person}{Thomas Prescher},
  \bibinfo{person}{Werner Haas}, \bibinfo{person}{Anders Fogh},
  \bibinfo{person}{Jann Horn}, \bibinfo{person}{Stefan Mangard},
  \bibinfo{person}{Paul~C. Kocher}, \bibinfo{person}{Daniel Genkin},
  \bibinfo{person}{Yuval Yarom}, {and} \bibinfo{person}{Michael Hamburg}.}
  \bibinfo{year}{2018}\natexlab{}.
\newblock \showarticletitle{Meltdown: Reading Kernel Memory from User Space}.
  In \bibinfo{booktitle}{\emph{USENIX Security Symposium}}.
\newblock


\bibitem[Liu et~al\mbox{.}(2015)]%
        {Liu2015LastLevelCS}
\bibfield{author}{\bibinfo{person}{Fangfei Liu}, \bibinfo{person}{Yuval Yarom},
  \bibinfo{person}{Qian Ge}, \bibinfo{person}{Gernot Heiser}, {and}
  \bibinfo{person}{Ruby~B. Lee}.} \bibinfo{year}{2015}\natexlab{}.
\newblock \showarticletitle{Last-Level Cache Side-Channel Attacks are
  Practical}. In \bibinfo{booktitle}{\emph{2015 IEEE Symposium on Security and
  Privacy}}. \bibinfo{publisher}{IEEE}, \bibinfo{pages}{605--622}.
\newblock
\urldef\tempurl%
\url{https://doi.org/10.1109/sp.2015.43}
\showDOI{\tempurl}


\bibitem[Lustig et~al\mbox{.}(2016)]%
        {Lustig2016COATCheckVM}
\bibfield{author}{\bibinfo{person}{Daniel Lustig}, \bibinfo{person}{Geet
  Sethi}, \bibinfo{person}{Margaret Martonosi}, {and} \bibinfo{person}{Abhishek
  Bhattacharjee}.} \bibinfo{year}{2016}\natexlab{}.
\newblock \showarticletitle{COATCheck: Verifying Memory Ordering at the
  Hardware-OS Interface}. In \bibinfo{booktitle}{\emph{Proceedings of the
  Twenty-First International Conference on Architectural Support for
  Programming Languages and Operating Systems}} (Atlanta, Georgia, USA)
  \emph{(\bibinfo{series}{ASPLOS '16})}. \bibinfo{publisher}{Association for
  Computing Machinery}, \bibinfo{address}{New York, NY, USA},
  \bibinfo{pages}{233--247}.
\newblock
\showISBNx{9781450340915}
\urldef\tempurl%
\url{https://doi.org/10.1145/2872362.2872399}
\showDOI{\tempurl}


\bibitem[Maisuradze and Rossow(2018)]%
        {Maisuradze2018ret2specSE}
\bibfield{author}{\bibinfo{person}{Giorgi Maisuradze} {and}
  \bibinfo{person}{Christian Rossow}.} \bibinfo{year}{2018}\natexlab{}.
\newblock \showarticletitle{ret2spec: Speculative Execution Using Return Stack
  Buffers}.
\newblock \bibinfo{journal}{\emph{Proceedings of the 2018 ACM SIGSAC Conference
  on Computer and Communications Security}} (\bibinfo{year}{2018}).
\newblock
\urldef\tempurl%
\url{https://api.semanticscholar.org/CorpusID:51804116}
\showURL{%
\tempurl}


\bibitem[McMillan(1993)]%
        {SymbolicMC}
\bibfield{author}{\bibinfo{person}{Kenneth~L. McMillan}.}
  \bibinfo{year}{1993}\natexlab{}.
\newblock \bibinfo{booktitle}{\emph{Symbolic model checking}}.
\newblock \bibinfo{publisher}{Kluwer}.
\newblock
\showISBNx{978-0-7923-9380-1}
\urldef\tempurl%
\url{https://doi.org/10.1007/978-1-4615-3190-6_3}
\showURL{%
\tempurl}


\bibitem[Moghimi(2023)]%
        {moghimi2023downfall}
\bibfield{author}{\bibinfo{person}{Daniel Moghimi}.}
  \bibinfo{year}{2023}\natexlab{}.
\newblock \showarticletitle{Downfall: Exploiting Speculative Data Gathering}.
  In \bibinfo{booktitle}{\emph{32nd USENIX Security Symposium (USENIX Security
  23)}}. \bibinfo{publisher}{USENIX Association}, \bibinfo{address}{Anaheim,
  CA}, \bibinfo{pages}{7179--7193}.
\newblock
\showISBNx{978-1-939133-37-3}
\urldef\tempurl%
\url{https://www.usenix.org/conference/usenixsecurity23/presentation/moghimi}
\showURL{%
\tempurl}


\bibitem[Mosier et~al\mbox{.}(2022)]%
        {Mosier2022AxiomaticHC}
\bibfield{author}{\bibinfo{person}{Nicholas Mosier}, \bibinfo{person}{Hanna
  Lachnitt}, \bibinfo{person}{Hamed Nemati}, {and} \bibinfo{person}{Caroline
  Trippel}.} \bibinfo{year}{2022}\natexlab{}.
\newblock \showarticletitle{Axiomatic hardware-software contracts for
  security}. In \bibinfo{booktitle}{\emph{Proceedings of the 49th Annual
  International Symposium on Computer Architecture}}
  \emph{(\bibinfo{series}{ISCA '22})}. \bibinfo{publisher}{ACM}.
\newblock
\urldef\tempurl%
\url{https://doi.org/10.1145/3470496.3527412}
\showDOI{\tempurl}


\bibitem[Mutlu et~al\mbox{.}(2005)]%
        {Mutlu2005OnRT}
\bibfield{author}{\bibinfo{person}{Onur Mutlu}, \bibinfo{person}{Hyesoon Kim},
  \bibinfo{person}{Jared Stark}, {and} \bibinfo{person}{Yale~N. Patt}.}
  \bibinfo{year}{2005}\natexlab{}.
\newblock \showarticletitle{On Reusing the Results of Pre-Executed Instructions
  in a Runahead Execution Processor}.
\newblock \bibinfo{journal}{\emph{IEEE Computer Architecture Letters}}
  \bibinfo{volume}{4} (\bibinfo{year}{2005}), \bibinfo{pages}{2--2}.
\newblock


\bibitem[Nemati et~al\mbox{.}(2020)]%
        {Nemati2020ValidationOA}
\bibfield{author}{\bibinfo{person}{Hamed Nemati}, \bibinfo{person}{Pablo
  Buiras}, \bibinfo{person}{Andreas Lindner}, \bibinfo{person}{Roberto
  Guanciale}, {and} \bibinfo{person}{Swen Jacobs}.}
  \bibinfo{year}{2020}\natexlab{}.
\newblock \showarticletitle{Validation of Abstract Side-Channel Models for
  Computer Architectures}.
\newblock \bibinfo{journal}{\emph{Computer Aided Verification}}
  \bibinfo{volume}{12224} (\bibinfo{year}{2020}), \bibinfo{pages}{225 -- 248}.
\newblock
\urldef\tempurl%
\url{https://api.semanticscholar.org/CorpusID:211567271}
\showURL{%
\tempurl}


\bibitem[Nilizadeh et~al\mbox{.}(2018)]%
        {Nilizadeh2018DifFuzzDF}
\bibfield{author}{\bibinfo{person}{Shirin Nilizadeh}, \bibinfo{person}{Yannic
  Noller}, {and} \bibinfo{person}{Corina~S. Pasareanu}.}
  \bibinfo{year}{2018}\natexlab{}.
\newblock \showarticletitle{DifFuzz: Differential Fuzzing for Side-Channel
  Analysis}.
\newblock \bibinfo{journal}{\emph{2019 IEEE/ACM 41st International Conference
  on Software Engineering (ICSE)}} (\bibinfo{year}{2018}),
  \bibinfo{pages}{176--187}.
\newblock


\bibitem[Noller and Tizpaz-Niari(2021)]%
        {Noller2021QFuzzQF}
\bibfield{author}{\bibinfo{person}{Yannic Noller} {and} \bibinfo{person}{Saeid
  Tizpaz-Niari}.} \bibinfo{year}{2021}\natexlab{}.
\newblock \showarticletitle{QFuzz: quantitative fuzzing for side channels}.
\newblock \bibinfo{journal}{\emph{Proceedings of the 30th ACM SIGSOFT
  International Symposium on Software Testing and Analysis}}
  (\bibinfo{year}{2021}).
\newblock


\bibitem[Oleksenko et~al\mbox{.}(2022)]%
        {Oleksenko2021RevizorTB}
\bibfield{author}{\bibinfo{person}{Oleksii Oleksenko},
  \bibinfo{person}{Christof Fetzer}, \bibinfo{person}{Boris Köpf}, {and}
  \bibinfo{person}{Mark Silberstein}.} \bibinfo{year}{2022}\natexlab{}.
\newblock \showarticletitle{Revizor: testing black-box CPUs against speculation
  contracts}.
\newblock  (\bibinfo{date}{Feb.} \bibinfo{year}{2022}).
\newblock
\urldef\tempurl%
\url{https://doi.org/10.1145/3503222.3507729}
\showDOI{\tempurl}


\bibitem[Oleksenko et~al\mbox{.}(2023)]%
        {oleksenko2023hide}
\bibfield{author}{\bibinfo{person}{Oleksii Oleksenko}, \bibinfo{person}{Marco
  Guarnieri}, \bibinfo{person}{Boris K\"{o}pf}, {and} \bibinfo{person}{Mark
  Silberstein}.} \bibinfo{year}{2023}\natexlab{}.
\newblock \showarticletitle{Hide and Seek with Spectres: Efficient discovery of
  speculative information leaks with random testing}. In
  \bibinfo{booktitle}{\emph{2023 IEEE Symposium on Security and Privacy (SP)}}.
  \bibinfo{publisher}{IEEE}.
\newblock
\urldef\tempurl%
\url{https://doi.org/10.1109/sp46215.2023.10179391}
\showDOI{\tempurl}


\bibitem[Polgreen et~al\mbox{.}(2022)]%
        {uclid5}
\bibfield{author}{\bibinfo{person}{Elizabeth Polgreen}, \bibinfo{person}{Kevin
  Cheang}, \bibinfo{person}{Pranav Gaddamadugu}, \bibinfo{person}{Adwait
  Godbole}, \bibinfo{person}{Kevin Laeufer}, \bibinfo{person}{Shaokai Lin},
  \bibinfo{person}{Yatin~A. Manerkar}, \bibinfo{person}{Federico Mora}, {and}
  \bibinfo{person}{Sanjit~A. Seshia}.} \bibinfo{year}{2022}\natexlab{}.
\newblock \showarticletitle{{UCLID5:} Multi-modal Formal Modeling,
  Verification, and Synthesis}. In \bibinfo{booktitle}{\emph{34th International
  Conference on Computer Aided Verification (CAV)}}
  \emph{(\bibinfo{series}{Lecture Notes in Computer Science},
  Vol.~\bibinfo{volume}{13371})}. \bibinfo{publisher}{Springer},
  \bibinfo{pages}{538--551}.
\newblock


\bibitem[Ponce-de Leon and Kinder(2022)]%
        {Len2021CatsVS}
\bibfield{author}{\bibinfo{person}{Hernán Ponce-de Leon} {and}
  \bibinfo{person}{Johannes Kinder}.} \bibinfo{year}{2022}\natexlab{}.
\newblock \showarticletitle{Cats vs. Spectre: An Axiomatic Approach to Modeling
  Speculative Execution Attacks}.
\newblock  (\bibinfo{date}{May} \bibinfo{year}{2022}).
\newblock
\urldef\tempurl%
\url{https://doi.org/10.1109/sp46214.2022.9833774}
\showDOI{\tempurl}


\bibitem[Rushby(1992)]%
        {Rushby1992NoninterferenceTT}
\bibfield{author}{\bibinfo{person}{John Rushby}.}
  \bibinfo{year}{1992}\natexlab{}.
\newblock \bibinfo{booktitle}{\emph{Noninterference, Transitivity, and
  Channel-Control Security Policies}}.
\newblock \bibinfo{type}{{T}echnical {R}eport}.
\newblock
\urldef\tempurl%
\url{http://www.csl.sri.com/papers/csl-92-2/}
\showURL{%
\tempurl}


\bibitem[Sanchez~Vicarte et~al\mbox{.}(2021)]%
        {Vicarte2021OpeningPB}
\bibfield{author}{\bibinfo{person}{Jose~Rodrigo Sanchez~Vicarte},
  \bibinfo{person}{Pradyumna Shome}, \bibinfo{person}{Nandeeka Nayak},
  \bibinfo{person}{Caroline Trippel}, \bibinfo{person}{Adam Morrison},
  \bibinfo{person}{David Kohlbrenner}, {and} \bibinfo{person}{Christopher~W.
  Fletcher}.} \bibinfo{year}{2021}\natexlab{}.
\newblock \showarticletitle{Opening Pandora's Box: A Systematic Study of New
  Ways Microarchitecture Can Leak Private Data}. In
  \bibinfo{booktitle}{\emph{2021 ACM/IEEE 48th Annual International Symposium
  on Computer Architecture (ISCA)}}. \bibinfo{publisher}{IEEE},
  \bibinfo{pages}{347--360}.
\newblock
\urldef\tempurl%
\url{https://doi.org/10.1109/isca52012.2021.00035}
\showDOI{\tempurl}


\bibitem[Schwarz et~al\mbox{.}(2019a)]%
        {Schwarz2019ZombieLoadCD}
\bibfield{author}{\bibinfo{person}{Michael Schwarz}, \bibinfo{person}{Moritz
  Lipp}, \bibinfo{person}{Daniel Moghimi}, \bibinfo{person}{Jo Van~Bulck},
  \bibinfo{person}{Julian Stecklina}, \bibinfo{person}{Thomas Prescher}, {and}
  \bibinfo{person}{Daniel Gruss}.} \bibinfo{year}{2019}\natexlab{a}.
\newblock \showarticletitle{ZombieLoad: Cross-Privilege-Boundary Data
  Sampling}. In \bibinfo{booktitle}{\emph{Proceedings of the 2019 ACM SIGSAC
  Conference on Computer and Communications Security}}
  \emph{(\bibinfo{series}{CCS '19})}. \bibinfo{publisher}{ACM}.
\newblock
\urldef\tempurl%
\url{https://doi.org/10.1145/3319535.3354252}
\showDOI{\tempurl}


\bibitem[Schwarz et~al\mbox{.}(2019b)]%
        {Schwarz2019NetSpectreRA}
\bibfield{author}{\bibinfo{person}{Michael Schwarz}, \bibinfo{person}{Martin
  Schwarzl}, \bibinfo{person}{Moritz Lipp}, {and} \bibinfo{person}{Daniel
  Gruss}.} \bibinfo{year}{2019}\natexlab{b}.
\newblock \showarticletitle{NetSpectre: Read Arbitrary Memory over Network}.
\newblock \bibinfo{journal}{\emph{ArXiv}}  \bibinfo{volume}{abs/1807.10535}
  (\bibinfo{year}{2019}).
\newblock


\bibitem[Sha et~al\mbox{.}(2006)]%
        {Sha2006NoSQSC}
\bibfield{author}{\bibinfo{person}{Tingting Sha}, \bibinfo{person}{Milo~M.K.
  Martin}, {and} \bibinfo{person}{Amir Roth}.} \bibinfo{year}{2006}\natexlab{}.
\newblock \showarticletitle{NoSQ: Store-Load Communication without a Store
  Queue}.
\newblock \bibinfo{journal}{\emph{2006 39th Annual IEEE/ACM International
  Symposium on Microarchitecture (MICRO'06)}} \bibinfo{volume}{27},
  \bibinfo{number}{1}, \bibinfo{pages}{285--296}.
\newblock
\showISSN{0272-1732}
\urldef\tempurl%
\url{https://doi.org/10.1109/mm.2007.17}
\showDOI{\tempurl}


\bibitem[Sha et~al\mbox{.}(2005)]%
        {Sha2005ScalableSF}
\bibfield{author}{\bibinfo{person}{Tingting Sha}, \bibinfo{person}{Milo M.~K.
  Martin}, {and} \bibinfo{person}{Amir Roth}.} \bibinfo{year}{2005}\natexlab{}.
\newblock \showarticletitle{Scalable store-load forwarding via store queue
  index prediction}.
\newblock \bibinfo{journal}{\emph{38th Annual IEEE/ACM International Symposium
  on Microarchitecture (MICRO'05)}} (\bibinfo{year}{2005}), \bibinfo{pages}{12
  pp.--170}.
\newblock


\bibitem[Sodani and Sohi(1997)]%
        {Sodani1997DynamicIR}
\bibfield{author}{\bibinfo{person}{Avinash Sodani} {and}
  \bibinfo{person}{Gurindar~S. Sohi}.} \bibinfo{year}{1997}\natexlab{}.
\newblock \showarticletitle{Dynamic Instruction Reuse}.
\newblock \bibinfo{journal}{\emph{Conference Proceedings. The 24th Annual
  International Symposium on Computer Architecture}} (\bibinfo{year}{1997}),
  \bibinfo{pages}{194--205}.
\newblock


\bibitem[Subramaniam and Loh(2006)]%
        {Subramaniam2006FireandForgetLS}
\bibfield{author}{\bibinfo{person}{Samantika Subramaniam} {and}
  \bibinfo{person}{Gabriel Loh}.} \bibinfo{year}{2006}\natexlab{}.
\newblock \showarticletitle{Fire-and-Forget: Load/Store Scheduling with No
  Store Queue at All}. In \bibinfo{booktitle}{\emph{2006 39th Annual IEEE/ACM
  International Symposium on Microarchitecture (MICRO'06)}}.
  \bibinfo{publisher}{IEEE}, \bibinfo{pages}{273--284}.
\newblock
\showISSN{1072-4451}
\urldef\tempurl%
\url{https://doi.org/10.1109/micro.2006.26}
\showDOI{\tempurl}


\bibitem[Terauchi and Aiken(2005)]%
        {Terauchi2005SecureIF}
\bibfield{author}{\bibinfo{person}{Tachio Terauchi} {and} \bibinfo{person}{Alex
  Aiken}.} \bibinfo{year}{2005}\natexlab{}.
\newblock \showarticletitle{Secure Information Flow as a Safety Problem}. In
  \bibinfo{booktitle}{\emph{Proceedings of the 12th International Conference on
  Static Analysis}} (London, UK) \emph{(\bibinfo{series}{SAS'05})}.
  \bibinfo{publisher}{Springer-Verlag}, \bibinfo{address}{Berlin, Heidelberg},
  \bibinfo{pages}{352--367}.
\newblock
\showISBNx{3540285849}
\urldef\tempurl%
\url{https://doi.org/10.1007/11547662_24}
\showDOI{\tempurl}


\bibitem[Trippel et~al\mbox{.}(2019)]%
        {Trippel2019SecurityVV}
\bibfield{author}{\bibinfo{person}{Caroline Trippel}, \bibinfo{person}{Daniel
  Lustig}, {and} \bibinfo{person}{Margaret Martonosi}.}
  \bibinfo{year}{2019}\natexlab{}.
\newblock \showarticletitle{{Security Verification via Automatic Hardware-Aware
  Exploit Synthesis: The CheckMate Approach}}.
\newblock \bibinfo{journal}{\emph{IEEE Micro}}  \bibinfo{volume}{39}
  (\bibinfo{year}{2019}), \bibinfo{pages}{84--93}.
\newblock


\bibitem[Tyagi et~al\mbox{.}(2022)]%
        {Tyagi2022TheHuzzIF}
\bibfield{author}{\bibinfo{person}{Aakash Tyagi}, \bibinfo{person}{Addison
  Crump}, \bibinfo{person}{Ahmad-Reza Sadeghi}, \bibinfo{person}{Garrett
  Persyn}, \bibinfo{person}{Jeyavijayan Rajendran}, \bibinfo{person}{Patrick
  Jauernig}, {and} \bibinfo{person}{Rahul Kande}.}
  \bibinfo{year}{2022}\natexlab{}.
\newblock \showarticletitle{TheHuzz: Instruction Fuzzing of Processors Using
  Golden-Reference Models for Finding Software-Exploitable Vulnerabilities}.
\newblock \bibinfo{journal}{\emph{ArXiv}}  \bibinfo{volume}{abs/2201.09941}
  (\bibinfo{year}{2022}).
\newblock


\bibitem[v.~Gleissenthall et~al\mbox{.}(2021)]%
        {Gleissenthall2021SolverAidedCH}
\bibfield{author}{\bibinfo{person}{Klaus v. Gleissenthall},
  \bibinfo{person}{Rami~G\"{o}khan K\i{}c\i{}}, \bibinfo{person}{Deian Stefan},
  {and} \bibinfo{person}{Ranjit Jhala}.} \bibinfo{year}{2021}\natexlab{}.
\newblock \showarticletitle{Solver-Aided Constant-Time Hardware Verification}.
  In \bibinfo{booktitle}{\emph{Proceedings of the 2021 ACM SIGSAC Conference on
  Computer and Communications Security}} (Virtual Event, Republic of Korea)
  \emph{(\bibinfo{series}{CCS '21})}. \bibinfo{publisher}{Association for
  Computing Machinery}, \bibinfo{address}{New York, NY, USA},
  \bibinfo{pages}{429--444}.
\newblock
\showISBNx{9781450384544}
\urldef\tempurl%
\url{https://doi.org/10.1145/3460120.3484810}
\showDOI{\tempurl}


\bibitem[van Schaik et~al\mbox{.}(2019)]%
        {Schaik2019RIDLRI}
\bibfield{author}{\bibinfo{person}{Stephan van Schaik}, \bibinfo{person}{Alyssa
  Milburn}, \bibinfo{person}{Sebastian Osterlund}, \bibinfo{person}{Pietro
  Frigo}, \bibinfo{person}{Giorgi Maisuradze}, \bibinfo{person}{Kaveh Razavi},
  \bibinfo{person}{Herbert Bos}, {and} \bibinfo{person}{Cristiano Giuffrida}.}
  \bibinfo{year}{2019}\natexlab{}.
\newblock \showarticletitle{RIDL: Rogue In-Flight Data Load}. In
  \bibinfo{booktitle}{\emph{2019 IEEE Symposium on Security and Privacy (SP)}}.
  \bibinfo{publisher}{IEEE}, \bibinfo{pages}{88--105}.
\newblock
\urldef\tempurl%
\url{https://doi.org/10.1109/sp.2019.00087}
\showDOI{\tempurl}


\bibitem[von Gleissenthall et~al\mbox{.}(2019)]%
        {Gleissenthall2019IODINEVC}
\bibfield{author}{\bibinfo{person}{Klaus von Gleissenthall},
  \bibinfo{person}{Rami~G{\"o}khan Kici}, \bibinfo{person}{Deian Stefan}, {and}
  \bibinfo{person}{Ranjit Jhala}.} \bibinfo{year}{2019}\natexlab{}.
\newblock \showarticletitle{IODINE: Verifying Constant-Time Execution of
  Hardware}.
\newblock \bibinfo{journal}{\emph{ArXiv}}  \bibinfo{volume}{abs/1910.03111}
  (\bibinfo{year}{2019}).
\newblock
\urldef\tempurl%
\url{https://api.semanticscholar.org/CorpusID:197672843}
\showURL{%
\tempurl}


\bibitem[Wang et~al\mbox{.}(2020)]%
        {Wang2019KLEESpectre}
\bibfield{author}{\bibinfo{person}{Guanhua Wang}, \bibinfo{person}{Sudipta
  Chattopadhyay}, \bibinfo{person}{Arnab~Kumar Biswas}, \bibinfo{person}{Tulika
  Mitra}, {and} \bibinfo{person}{Abhik Roychoudhury}.}
  \bibinfo{year}{2020}\natexlab{}.
\newblock \showarticletitle{KLEESpectre: Detecting Information Leakage through
  Speculative Cache Attacks via Symbolic Execution}.
\newblock \bibinfo{journal}{\emph{ACM Transactions on Software Engineering and
  Methodology}} \bibinfo{volume}{29}, \bibinfo{number}{3} (\bibinfo{date}{June}
  \bibinfo{year}{2020}), \bibinfo{pages}{1--31}.
\newblock
\showISSN{1557-7392}
\urldef\tempurl%
\url{https://doi.org/10.1145/3385897}
\showDOI{\tempurl}


\bibitem[Wang et~al\mbox{.}(2023)]%
        {wang2023specification}
\bibfield{author}{\bibinfo{person}{Zilong Wang}, \bibinfo{person}{Gideon Mohr},
  \bibinfo{person}{Klaus von Gleissenthall}, \bibinfo{person}{Jan Reineke},
  {and} \bibinfo{person}{Marco Guarnieri}.} \bibinfo{year}{2023}\natexlab{}.
\newblock \showarticletitle{Specification and Verification of Side-channel
  Security for Open-source Processors via Leakage Contracts}. In
  \bibinfo{booktitle}{\emph{Proceedings of the 2023 ACM SIGSAC Conference on
  Computer and Communications Security}} (, Copenhagen, Denmark,)
  \emph{(\bibinfo{series}{CCS '23})}. \bibinfo{publisher}{Association for
  Computing Machinery}, \bibinfo{address}{New York, NY, USA},
  \bibinfo{pages}{2128--2142}.
\newblock
\showISBNx{9798400700507}
\urldef\tempurl%
\url{https://doi.org/10.1145/3576915.3623192}
\showDOI{\tempurl}


\bibitem[Weber et~al\mbox{.}(2021)]%
        {Weber2021OsirisAD}
\bibfield{author}{\bibinfo{person}{Daniel Weber}, \bibinfo{person}{Ahmad
  Ibrahim}, \bibinfo{person}{Hamed Nemati}, \bibinfo{person}{Michael Schwarz},
  {and} \bibinfo{person}{Christian Rossow}.} \bibinfo{year}{2021}\natexlab{}.
\newblock \showarticletitle{Osiris: Automated Discovery of Microarchitectural
  Side Channels}. In \bibinfo{booktitle}{\emph{30th USENIX Security Symposium
  (USENIX Security 21)}}. \bibinfo{publisher}{USENIX Association},
  \bibinfo{pages}{1415--1432}.
\newblock
\showISBNx{978-1-939133-24-3}
\urldef\tempurl%
\url{https://www.usenix.org/conference/usenixsecurity21/presentation/weber}
\showURL{%
\tempurl}


\bibitem[Zeng et~al\mbox{.}(2022)]%
        {Zeng2022AutomaticGO}
\bibfield{author}{\bibinfo{person}{Yu Zeng}, \bibinfo{person}{Aarti Gupta},
  {and} \bibinfo{person}{Sharad Malik}.} \bibinfo{year}{2022}\natexlab{}.
\newblock \showarticletitle{Automatic generation of architecture-level models
  from RTL designs for processors and accelerators}. In
  \bibinfo{booktitle}{\emph{Proceedings of the 2022 Conference \& Exhibition on
  Design, Automation \& Test in Europe}} (Antwerp, Belgium)
  \emph{(\bibinfo{series}{DATE '22})}. \bibinfo{publisher}{European Design and
  Automation Association}, \bibinfo{address}{Leuven, BEL},
  \bibinfo{pages}{460--465}.
\newblock
\showISBNx{9783981926361}


\end{thebibliography}

\newpage
\appendix

\section{Appendix}

We provide the complete versions of the algorithms and proofs.

\subsection{Template Generation}

\subsubsection{\gentemp{} Algorithm}

We provide the full description of the template generation procedure in Algorithm \ref{alg:gen-templates-app}.
This algorithm generates a set of pattern templates up to a certain depth $d$ given 
a semantic platform definition $M$, a non-interference property $\mathsf{NI}(\srcvars, \obsvars, \initialpred)$, and depth $d \in \mathbb{N}$.

As discussed in the main text, the algorithm iterates over all templates, first checking if the template propagates taint from the source variables to the observed variables, and if so, performing a semantic analysis to check whether the template violates the non-interference property.

\setcounter{algocf}{0}

\begin{algorithm}
    \small
    \DontPrintSemicolon
    \SetKwFunction{TemplateHelper}{TemplateHelper}
    \SetKwProg{Fn}{Function}{:}{}
    
    \KwInput{Semantic platform definition $M$, non-interference property $\mathsf{NI}(\srcvars, \obsvars, \initialpred)$, depth $d \in \mathbb{N}$}
    \KwOutput{A set of pattern templates}
    \KwData{acc: the accumulated set of pattern templates}

    \Fn{\TemplateHelper{$\askeleton$}}{
        \tcc{Search depth not reached?}
        \If{$|\askeleton| < d$}
        {
            \For{$\anop \in \setops$}
            {
                \If{$\anop\cdot \askeleton$ propagates taint from $\srcvars^C$ to $\obsvars$}
                {
                    \lIf{$\anop\cdot \askeleton \not\models \mathsf{NI}$}
                    {
                        acc.append($\anop\cdot \askeleton$)
                    }
                }
                \TemplateHelper{$\anop\cdot\askeleton$}    
            }
        }
    }
    \tcc{Search over depth $d$ templates}
    \TemplateHelper{$\epsilon$} \;
    \Return acc
\caption{\gentemp($M$, $\mathsf{NI}$, $d$)}
\label{alg:gen-templates-app}
\end{algorithm}

\newcommand{\atom}[1]{\mathsf{atom}_{#1}}

\begin{algorithm*}[h]
\SetKwFunction{GrammarHelper}{ConsHelper}
\SetKwProg{Fn}{Function}{:}{}
\DontPrintSemicolon
    \KwInput{Semantic platform definition $M$, non-interference property $\mathsf{NI}$, pattern template $\askeleton$, and a grammar $G$}
    \KwOutput{A set of patterns}
    \KwData{acc: an accumulated set of patterns}
    
    \Fn{\GrammarHelper{$\apred, i, L$}}{
        \If{$i > |L|$}{
            acc.append($(\askeleton, \apred)$)  \tcc*[r]{Exhausted all atomic predicates?}
        }
        \Else{
            $\atom{0} = L[i]$; $\atom{1} = L[i+1]$; $\atom{2} = L[i+2]$ \tcc*[r]{Choose atoms to branch on}
            \tcc{Does adding $\neg \atom{0}$, i.e., $(= \neg L[i])$ eliminate violations?}
            \lIf{
                $\forall \aninstr_1, \cdots, \aninstr_{|\askeleton|}.~  
                \aninstr_1 \cdots\aninstr_{|\askeleton|} \models (\askeleton, \apred \land \neg \atom{0}) \implies
                \aninstr_1, \cdots, \aninstr_{|\askeleton|} \models \mathsf{NI}$ \;
            }
            { \GrammarHelper{$\apred \land L[i]$, $i+1, L$} \tcc*[f]{add $L[i]$}}
            \tcc{Does adding $\neg \atom{0} \land \neg \atom{1}$ eliminate all violations?}
            {\color{blue}
            \ElseIf{
                $\forall \aninstr_1, \cdots, \aninstr_{|\askeleton|}.~  
                \aninstr_1 \cdots\aninstr_{|\askeleton|} \models (\askeleton, \apred \land \neg \atom{0} \land \neg \atom{1}) \implies
                \aninstr_1, \cdots, \aninstr_{|\askeleton|} \models \mathsf{NI}$ \;
            }
            { 
                \tcc{add $L[i], L[i+1]$ disjuncts (multi-counterfactual branching)}
                \GrammarHelper{$\apred \land L[i]$, $i+2, L$}; \GrammarHelper{$\apred \land L[i+1]$, $i+2, L$};
            }
            \tcc{Does adding $\neg \atom{0} \land \neg \atom{1} \land \neg \atom{2}$ eliminate all violations?}
            \ElseIf{
                $\forall \aninstr_1, \cdots, \aninstr_{|\askeleton|}.~  
                \aninstr_1 \cdots\aninstr_{|\askeleton|} \models (\askeleton, \apred \land \neg \atom{0} \land \neg \atom{1} \land \neg \atom{2}) \implies
                \aninstr_1, \cdots, \aninstr_{|\askeleton|} \models \mathsf{NI}$ \;
            }
            { 
                \tcc{add $L[i], L[i+1], L[i+2]$ disjuncts (multi-counterfactual branching)}
                \GrammarHelper{$\apred \land L[i]$, $i+3, L$}; \GrammarHelper{$\apred \land L[i+1]$, $i+3, L$}; \GrammarHelper{$\apred \land L[i+2]$, $i+3, L$};
            }}
            \lElse(\tcc*[f]{skip over $L[i]$}){
                \GrammarHelper{$\apred, i+1, L$};
            }
        }
    }
    $L$ = ApplyPredicates($\askeleton, G$) \tcc*[r]{Create all atoms in $L$}
    \GrammarHelper{$\texttt{true}, 0, L$} \tcc*[r]{Counterfact. addition}
    \Return acc
\caption{\conspec($M, \mathsf{NI}, \askeleton, G$)}
\label{alg:gram-search-full}
\end{algorithm*}

\subsubsection{Proof of Lemma \ref{lem:templategen}}

The proof of Lemma \ref{lem:templategen} is a direct consequence of the algorithm \gentemp{} and the soundness of a taint-based overapproximation.
We rely on the fact that, if a template $\askeleton$ does not propagate taint from $\srcvars$ to $\obsvars$, then it also does not violate the non-interference property $\mathsf{NI}(\srcvars, \obsvars, \initialpred)$.

\subsection{Constraint-based Template Specialization}

\subsubsection{Counterfactual-based Addition}

Counterfactual-based addition is based on Observations 1 and 2 which can be shown through simple propositional logic.
We will only prove Observation 2 since it is a generalization of Observation 1.

\begin{proof}[Proof of Observation \ref{obs:multicounterfactual}]
    Consider a partially specialized template $(\askeleton, \apred)$, and a set of atoms $\{\afacet_i\}_i$. Let $S$ be the set of programs:
    $S = \{ \aprogram ~|~ \aprogram \models (\askeleton, \apred)  \}$. Also for each $\afacet_i$, let $S_i$ be the set of programs:
    $S_i = \{ \aprogram ~|~ \aprogram \models (\askeleton, \apred \land \afacet_i) \}$. Finally, let $T$ be the programs: $T = \{ \aprogram ~|~ \aprogram \models \mathsf{NI} \}$.
    Then, the set $U$ of programs satisfying 
    \begin{equation*}
        \aprogram \models (\askeleton, \apred \land 
        \bigwedge_i \neg \afacet_i)
    \end{equation*}
    can we described as $U = S \setminus \bigcup_i S_i$.

    Now, suppose to the contrary that the statement was not true.
    Then there exists a program $\aprogram$, such that the following holds:
    \begin{equation}
        \aprogram \models (\askeleton, \apred) \land \aprogram \not\models \mathsf{NI}  \quad \text{ i.e., } \aprogram \in S \cap \overline{T}
    \end{equation}
    and
    \begin{equation}
        \neg\bigvee_i (\aprogram \models (\askeleton, \apred \land \afacet_i)) \quad \text{i.e., } \aprogram \not\in \bigcup_i S_i
    \end{equation}

    Then, by (3) and (4), $\aprogram \in (S \cap \overline{T} \cap \overline{\bigcup_i S_i})$, which implies, by the definition of $U$, that $\aprogram \in U \cap \overline{T}$.
    However, this means that the antecedent of the observation (which corresponds to $\overline{U} \cup T$) does not infact hold, which is a contradiction.
\end{proof}

\subsubsection{Full \conspec{} Algorithm}

We provide the full description of the constraint-based template specialization procedure (Algorithm \ref{alg:gram-search-full}).

\subsubsection{$k$-completeness of \conspec{}}

Theorem \ref{thm:kcomplete} is proven by induction on
the number of specialization iterations.

\begin{proof}[Proof of Theorem \ref{thm:kcomplete}]
Suppose \conspec{} is invoked with arguments $(\askeleton, \mathsf{NI}, M, G)$.
We claim that in Alg. \ref{alg:gram-search-full}, at each recursive call of $\mathsf{ConsHelper}$ with arguments $(\apred, i, L)$, the following property holds:
\begin{align*}
    &\forall \aprogram.~~(\aprogram \models (\askeleton, \mathsf{true}) \land C \not\models \mathsf{NI}) \\
    &\quad \implies
    \exists (\askeleton, \apred') \in \text{acc} \cup \{(\askeleton, \apred)\}.~ 
    \aprogram \models (\askeleton, \apred')
\end{align*}
In words, the partially specialized patterns maintained in the accumulator queue form an overapproxiation of the violating programs for skeleton $\askeleton$.

\textbf{Base case:} At the start, $\mathsf{ConsHelper}(\mathsf{true}, 0, L)$ immediately implies the property. 

\textbf{Inductive step:} Let the property hold for $\mathsf{ConsHelper}(\apred, i, L)$. Then, either we (a) add $(\askeleton, \apred)$ to acc, (b) we skip over $L[i]$ or (c) reinvoke $\mathsf{ConsHelper}$ with an incremented $i$.
The property holds immediately in cases (a) and (c). In case (b), we note that 
the set of programs $C$ that satisfying $(\askeleton, \apred)$ and violating $\mathsf{NI}$ is identical to those that violate $(\askeleton, \apred \land L[i])$
due to Observation 1. A similar argument holds for the multi-counterfactual case (using Obs. 2).

Therefore, by induction, the property holds for all recursive calls of $\mathsf{ConsHelper}$, and hence, the final set of patterns in \texttt{acc} (when there are no more calls to $\mathsf{ConsHelper}$) is a $k$-complete set of patterns for $\askeleton$.
\end{proof}

\end{document}